\documentclass[11pt]{article}
\usepackage{amssymb,amsmath}
\usepackage{epsfig}
\usepackage{palatino}
\usepackage{graphicx}
\usepackage{textcomp}
\usepackage{hyperref}

\newtheorem{theorem}{Theorem}

\newtheorem{definition}{Definition}

\newtheorem{claim}{Claim}
\newtheorem{lemma}{Lemma}

\newtheorem{fact}{Fact}

 \newcommand{\qedsymb}{\hfill{\rule{2mm}{2mm}}}  
 \newenvironment{proof}[1][]{\begin{trivlist}  
 \item[\hspace{\labelsep}{\bf\noindent Proof#1:\/}] 
 }{\qedsymb\end{trivlist}}

\def\ra{\rangle}

\newcommand{\be}{\begin{eqnarray}}
\newcommand{\ee}{\end{eqnarray}}

\newcommand\floor[1]{{\lfloor #1 \rfloor}}
\newcommand\ceil[1]{{\lceil #1 \rceil}}

\newcommand\ket[1]{{ |{#1} \rangle }}
\newcommand\bra[1]{{ \langle {#1} | }}

\def\PCP{{\sf{PCP}}}
\def\PCPP{{\sf{PCPP}}}
\def\qPCP{{\sf{qPCP}}}
\def\qPCPP{{\sf{qPCPP}}}
\def\NLTS{{\sf{NLTS}}}
\def\NP{{\sf{NP}}}

\def\P{{\sf{P}}}

\def\qLTC{{\sf{qLTC}}}
\def\LTC{{\sf{LTC}}}
\def\sLTC{{\sf{sLTC}}}

\def\CSS{{\sf{CSS}}}

\newcommand{\ignore}[1]{}
\newcommand{\dnote}[1]{}
\newcommand{\lnote}[1]{}

\newcommand{\eps}{\varepsilon}
\renewcommand{\epsilon}{\varepsilon}

\title{Quantum Locally Testable Codes}
\begin{document}

\author{Dorit Aharonov\thanks{School of Computer Science and
Engineering, The Hebrew University,
Jerusalem, Israel}
 \and Lior Eldar\thanks{School of Computer Science and
Engineering, The Hebrew University,
Jerusalem, Israel.}}

\date{\today}

\maketitle


\begin{abstract}
We initiate the study of quantum Locally Testable
Codes ($\qLTC$s). We provide a definition
together with a simplification, denoted $\sLTC$s, for the special case of
stabilizer codes,
and provide some basic results using those definitions. 
The most crucial parameter of such codes is their {\it soundness}, 
$R(\delta)$, namely, the probability that a randomly chosen 
constraint is violated as a function of the distance of a word from the code 
($\delta$, the relative distance from the code, is called the {\it proximity}). 
We then proceed to study limitations on 
$\qLTC$s. In our first main result we prove a 
surprising, inherently quantum, property of $\sLTC$s: 
for small values of proximity, 
the better the small-set expansion of the
interaction graph of the constraints, 
the {\it less} sound the $\qLTC$ becomes. 
This stands in sharp contrast to the classical setting. 
The complementary, more intuitive, result also holds: 
an upper bound on the soundness when the 
code is defined on {\it bad} 
small-set expanders (a bound which turns out to be far more difficult to 
show in the quantum case). 
Together we arrive at a quantum upper-bound 
on the soundness of stabilizer $\qLTC$s set on 
{\it any} graph, which does not hold in the classical case.    
Many open questions are raised regarding what possible parameters 
are achievable for $\qLTC$s.  
In the appendix we also define a quantum analogue 
of $\PCP$s of proximity ($\PCPP$s) and point out that the 
result of \cite{BGHSV} by which $\PCPP$s imply 
$\LTC$s with related parameters, carries over to the $\sLTC$s. 
This creates a first link between $\qLTC$s and quantum $\PCP$s \cite{AAV}.  
\end{abstract}

\section{Introduction} 
Quantum error correcting codes have played a crucial role in
quantum complexity theory 
(see, e.g., \cite{Got2,Got3,BC, BCG})
and their study is a vastly growing field
(see, e.g.,\cite{Got,Zemor2,Zemor3, Kovalev1,
Kovalev2,TerhalTradeoffs,Fetaya}); they are related to 
a variety of issues including resilience to noise
and fault tolerance, quantum cryptography, topological order,    
multi-particle entanglement, and more. 

Here, we initiate the study of the quantum analogue of Locally Testable 
Codes ($\LTC$s). $\LTC$s, first defined in \cite{FS,RS,A}, 
are a particularly interesting class of error correcting codes which 
played an instrumental role in all proofs of the celebrated
$\PCP$ theorem \cite{ALMSS, AS, Din}; their study had inspired the 
definition of property testing \cite{GGR} 
and the understanding of their limitations and possible constructions 
has developed into a very interesting field of its own 
(see for example Goldreich's survey \cite{Gold}).  
 
To define $\LTC$s, consider the following question: 
given a code of $n$-bit strings, 
defined by $O(1)$-local constraints, 
and a word which is of distance $\delta n>0$ from 
the code (we say it has {\it proximity} $\delta$), 
what is the probability that a randomly chosen
constraint is violated? 
We denote by $R(\delta)$ (called the {\it soundness})
the lower bound on the probability 
that {\it any} word of proximity $\delta$ from 
the code will violate a randomly 
chosen constraint. 

$\LTC$s of excellent soundness at proximities larger than some 
constant are known, most notably the Reed-Muller code \cite{Mac}, 
the Hadamard code \cite{AB}, and Hastad's 
long-code \cite{Hastad} which
were used in the $\PCP$ proofs of \cite{ALMSS, AS,Din}.
Though of excellent soundness, these codes are not so satisfying when 
considering other parameters of interest. For example, the rates of 
the Hadamard 
and long code are exponentially and doubly exponentially small, 
respectively. Much research \cite{FS,RS} was devoted to optimizing 
the parameters of $\LTC$s, maintaining constant relative distance 
and constant {\it query complexity}
(namely, the number of bits in each constraint), and improving the rate.   
The best known $\LTC$s in this respect are \cite{Din,BS} which
have constant distance, constant query complexity, 
and rates which are $1/polylog$.  
It is a major open question (called the $c^3$ problem \cite{GS}) 
whether good (namely, constant relative rate and distance)
$\LTC$s exist. 

\subsection{Quantum Locally Testable Codes - Definition and Motivation} 
To the best of our knowledge the quantum analogue of 
$\LTC$s was not defined before. 
We provide a definition of
general quantum Locally Testable Codes ($\qLTC$s) 
in Definition \ref{def:qLTC}. 
To define $\qLTC$s, we recall that a 
quantum code defined by $O(1)$-local constraints can be viewed 
as the groundspace (namely, the zero eigenspace) 
of a local Hamiltonian $H=\sum_{i=1}^m \Pi_i$ 
whose local terms are 
projections, which we will refer to as the quantum constraints.  
We define Quantum Locally Testable Codes ($\qLTC$s) with soundness 
$R(\delta)$ as 
those codes for which 
when a state $\Psi$ is within distance 
at least $\delta n$ from the code 
space, its average energy with respect to the constraints, 
$\frac{1}{m}\langle \psi|H|\psi\rangle$, 
is at least $R(\delta)$ 
(for an exact definition see Subsection \ref{sec:qLTC}). 
The average energy is the natural and commonly used 
analogue, in quantum Hamiltonian complexity, 
of the probability to detect a violation in a randomly chosen 
constraint (see for example \cite{AAV}). 

This definition sets the stage for a wide range of 
interesting questions.
What are the limitations on quantum $\LTC$s, and what are possible 
constructions?  
Are there $\qLTC$s which achieve, or get close to, 
the best classical $\LTC$s in terms of parameters, 
or are the quantum versions of those codes inherently limited by some quantum 
phenomenon? What can we learn from $\qLTC$s regarding the notion 
of local testability of proofs, a notion which in the classical setting 
is tightly related to that of $\LTC$s \cite{Gold}, 
and which is still widely 
evasive in the quantum setting \cite{AAV}?  

Our motivation in introducing $\qLTC$s 
in order to study the above questions 
stems not only from trying to import the interesting 
classical local-testability paradigm into the quantum setting, but 
also from their strong relations to 
questions which are of inherent 
interest to quantum information, quantum complexity as well as to 
quantum physics. We highlight here several such connections. 

An important motivation is to gain insight into the widely open quantum 
$\PCP$ conjecture \cite{Aha}, 
a quantum analogue of the $\PCP$ theorem; it states, roughly,  
that it is quantum-$\NP$ hard 
to approximate the ground energies
of local Hamiltonians even to within 
a constant fraction. This conjecture is tightly related to 
deep questions about 
multiparticle entanglement, and 
there has been much recent work attempting to make progress on it 
(see the recent survey \cite{AAV} and references therein). 
In the classical setting, 
$\LTC$s have been instrumental in $\PCP$ theory \cite{ALMSS, AS, Din}  
and are intimately related to the 
notion of local testability of proofs 
\cite{Gold}, and  
understanding the 
limitations of their quantum counterparts
might shed light on the $\qPCP$ problem. 

Another important open question is that of 
the feasibility of quantum self correcting memory. This is   
a medium in which 
a quantum state is maintained almost in tact 
for a long time without active error correction, 
even at constant temperatures; 
errors are corrected passively by the interaction with the environment.  
Clearly such a system is of high practical as well as theoretical interest,  
and the topic has been studied extensively in recent years 
(e.g., \cite{BT,Den,Ha,HaahPreskill,CL, Haah, Yoshida}). 
It is of major interest to devise feasible constructions of 
quantum self correcting memory.
A crucial role in this area is played by 
the {\it energy barrier} of the quantum code, which
is the amount of energy required in order to 
move from one codeword to an orthogonal one. 
This notion, which has also been studied extensively (see, e.g., 
\cite{M}), is tightly related to the soundness of the code, 
which can be viewed as the {\it energy cost} of 
large errors; understanding $\qLTC$s might thus provide insights 
into possible constructions of self-correcting memories.  

A fundamental open question related to both of the above is
whether multiparticle entanglement can be made robust 
at room temperatures. The question was formalized by Hastings in terms of  
the $\NLTS$ (No Low-energy Trivial States)
conjecture \cite{Has}, which, roughly, states 
that there exist local Hamiltonians 
such that all their low-energy states are highly entangled. 
Such Hamiltonians are necessary for the $\qPCP$ conjecture to 
hold (\cite{Has}, and see also \cite{AAV}). 
$\NLTS$ Hamiltonians and $\qLTC$s seem related: 
while in $\qLTC$s, low energies imply 
closeness to the code, in $\NLTS$ Hamiltonians they imply high entanglement, 
which is well known to be necessary for code states. 
Indeed, some weak connections between the two notions were already proven
\footnote{One can show that $\qLTC$s do not have tensor-product states with small (constant) mean energy.}.

In the following, we will investigate the behavior of $\qLTC$s in 
various scenarios.  
The behavior of $\LTC$s is usually explored in one of two contexts:
as an error-correcting code, or in relation to locally testable 
proofs (see \cite{Gold}); depending on the context, one is interested 
in different parameters. In particular, in the context of error correction, 
the interesting regime of proximities, namely distance of the word from the 
code, is at most half the distance {\it of} the code; in this regime, 
the error can still be corrected. In the context of $\PCP$s, on the other 
hand, much larger distances can be of interest, since a cheating prover may  
provide witnesses of arbitrary distance from the code. 
At any given point, we will mention the range which we will be 
considering. 

\subsection{Contributions}
\subsubsection{Definition and Basic Examples}
We provide a general definition of $\qLTC$s in Definition \ref{def:qLTC}. 
Being probably the richest and most well-studied class of quantum 
codes, stabilizer codes \cite{Got} are compelling to work with.  
We thus provide a simpler definition for stabilizer $\LTC$s (denoted 
$\sLTC$ -- Definition \ref{def:sLTC}) and  
prove that it coincides with the definition of $\qLTC$s on stabilizer codes, 
in Claim \ref{cl:sLTCqLTC}.

An illuminating example to consider is Kitaev's $2D$
toric code \cite{Kit2}, which turns out to have very bad soundness, 
since a string-like error of any length -- e.g.,  
an error made of Pauli operators applied on a $\Theta(\sqrt{n})$ long 
line-segment of qubits --  
only violates two constraints - those that intersect its two edges. 
So, at small (up to $1/\sqrt{n}$) values of proximity, 
the soundness is bounded from above by $1/\sqrt{n}$.  
One can in fact extend this phenomenon to derive bounds on the soundness 
for constant values of proximity. 

Another illuminating example is the Quantum Reed-Muller codes 
\cite{Ste}. Certain classical Reed-Muller code are known 
to have good (constant) soundness \cite{Alon}. 
Quantum Reed-Muller codes can be constructed 
using classical Reed-Muller codes and their dual, 
in the usual $\CSS$ paradigm \cite{Nielsen}. 
By construction, the resulting code will  
inherit its soundness from one of the two classical codes that defines the 
$\CSS$ code -- the one with the worse soundness.   
Unfortunately, the rate and distance of the quantum Reed-Muller codes  
are much worse even than the optimal classical Reed-Muller codes,  
as is expected from $\CSS$ codes \cite{Nielsen}
(more details will be provided in the journal version).  

\subsubsection{Bound on the soundness of $\sLTC$s on small set expanders} 
We provide two upper bounds 
on the soundness of $\qLTC$s at low, constant
values of proximities $\delta>0$. 
We focus on $\sLTC$'s on $n$ qudits, 
which are {\it good} quantum codes,
defined by $m$ $k=O(1)$-local check terms, where
each qudit participates in $D_L=O(1)$ constraints.  
For such codes, we consider bounds on the soundness at values of 
proximities which are at most some constant; 
this constant is a function of $k,D_L$, and 
in particular, $\delta<1/k$. 
Usually, in the classical setting, it is much easier to derive $\LTC$s 
whose soundness is good (large) for those {\it small} proximity values.  
Here, we show that in this supposedly {\it easier} range of 
parameters, $\qLTC$s are severely limited compared to their classical 
counterparts.  

To make the statement of the results simpler, we observe that 
the soundness $R(\delta)$, is bounded above by 
the number of constraints that touch the erred qudits, divided by $m$:   
hence it is at most $\delta n D_L/m=k\delta$ 
(using $D_L n =km$). It is more informative 
to present our results in terms of the {\it relative soundness} 
$r(\delta)= R(\delta)/k\delta$, which is the soundness normalized 
by its maximal value
(for exact definition see Definition \ref{def:smallr}).

Our first main result proves that good $\qLTC$s exhibit a severe 
limitation on their relative soundness, when set on good expanders. 
More precisely, consider the bi-partite graph of the code 
defined with $n$ bits on the left side,  
$m$ constraints on the other side, and edges connecting each 
constraint to all of its bits.
We say that the bi-partite graph is an $\epsilon$ small-set expander 
if every small (size $k=O(1)$) subset of bits, is examined nearly by as many 
constraints as it possibly can, namely, by at least $(1-\epsilon)kD_L$ 
constraints. 
Theorem \ref{thm:QECC} shows that in the quantum setting, 
when the underlying bi-partite graph of the $\sLTC$ code 
is an $\epsilon$ small set expander, the relative  
soundness is $O(\epsilon)$. In other words, the better the expansion, 
the worse the soundness. This holds  
for all proximities smaller than some constant $\delta_0$.  
More formally, we show:
\begin{theorem}\label{thm:QECC}  
Let $C$ be a good stabilizer code, on $n$ $d$-dimensional qudits, of
relative distance $>0$, and a $k$-local 
generating set ${\cal G}\subset \Pi_d^n$,
such that each qudit is examined by $D_L$ generators.
Put $\delta_0 = \min \left\{\frac{1}{k^3 \cdot D_L},\frac{1}{2n}dist(C)\right\}$.
Suppose the bi-partite interaction graph of ${\cal G}$ is $\eps$-small 
set expanding, 
for $\eps < 1/2$. 
Then,  for all $0<\delta < \delta_0$, 
we have $r(\delta)\le 2\eps$.
\end{theorem}
See subsection \ref{sec:qLTC} for exact definitions of 
Stabilizer codes and their generators, and Definition 
\ref{def:smallr} for the exact 
definition of relative soundness. 
   
Theorem \ref{thm:QECC} stands in sharp contrast to    
the classical domain. Classically, 
codes can easily be constructed on good expanders 
so that for small proximities their soundness is excellent;  
We provide an explicit  
such example whose relative soundness is arbitrarily close to 
$1$ by plugging the {\it lossless expanders} constructed in \cite{CRVW}, into 
the expander code construction of Sipser and Spielman \cite{Spi}. This 
implies good classical codes with 
constant query complexity and with almost optimal 
soundness for any proximity $\delta$ smaller than some constant  
(see Claim (\ref{cl:classical} in Appendix \ref{app:classical}). 

\subsubsection{Bound on the soundness of general $\qLTC$s} 
Our second main result is an upper bound on the relative soundness which 
holds for $\sLTC$s set on {\it any} underlying bi-partite graph, 
not necessarily small-set expanders.  

\begin{theorem}\label{thm:sound} (Roughly)
For any good stabilizer code $C$
of $k$-local terms ($k\geq 4$) over $d$-dimensional qudits, where
each qudit interacts with $O(1)$ 
local terms,
errors of fractional weight $\delta<\delta_0\le 1$, for 
$\delta_0 = \Omega(1)$
have relative soundness at most 
 $\alpha(d)(1-\gamma_{gap})$ 
for some constant function $\gamma_{gap}=\gamma_{gap}(k,d)>0$. 
\end{theorem} 

\noindent
$\alpha(d)$ in the above theorem is defined to be $1-1/d^2$; 
this is a technical upper bound 
on the relative soundness 
of $\qLTC$s defined on $d$-dimensional qudits,
stemming quite easily from the size of the alphabet $d$
(see subsection \ref{sec:alphabet}); 
Theorem \ref{thm:sound} shows
that the soundness is further bounded by some seemingly deeper quantum 
phenomenon. 
We stress that this upperbound, which is not exhibited in classical 
codes, is found in the range of parameters of 
$\delta$ (small constants) in which it 
is supposed to be {\it easiest} to achieve 
soundness for $\LTC$s, e.g., our 
Claim \ref{cl:classical}.  

\subsubsection{Quantum $\PCP$s of Proximity}
$\LTC$s are tightly connected \cite{Gold} to 
$\PCP$'s of proximity ($\PCPP$s), which 
are proof systems defined very similarly to 
$\PCP$s (See \cite{BGHSV}). 
For the reader familiar with $\PCP$s, they too consider a verifier 
who gets access to an untrusted proof, however, $\PCPP$s differ from 
$\PCP$s in two important aspects: 
first, they are weaker, in the sense that they 
are required to reject only inputs that are {\it far} from the language, 
whereas in $\PCP$s any input out of the language should be rejected. 
On the other hand, the verifier is charged not only for the number of 
queries out of the proof, 
but also for the number queries out of (part of) the input. 
For a formal definition see Appendix \ref{app:PCPP}.  

Ben Sasson et. al \cite{BGHSV} provide a standard construction of 
an $\LTC$ from a $\PCPP$. Given a $\PCPP$ for membership in a code, 
and an error correcting code $C$, 
they construct an $\LTC$ code $C'$, which inherits its soundness parameter 
from the soundness parameter of the $\PCPP$ and its distance from the 
code $C$ (Construction 4.3, and Proposition 4.4
in \cite{BGHSV}, see Appendix \ref{app:PCPP}). 

In Appendix \ref{app:PCPP}, we suggest a definition of quantum $\PCPP$s, 
and show that a similar result to that of \cite{BGHSV} 
holds in the quantum setting. 
The meaning 
of the definition of $\qPCPP$ and of the above described connection, 
and their relevance and importance to the quantum $\PCP$ conjecture,  
are far from clear (see for example \cite{AAV} for doubts 
regarding the classical approach to proving the quantum $\PCP$ conjecture, 
and the direct applicability of quantum Error correcting codes in this 
context). Still we provide these definitions and results in the appendix, 
to make the point that a syntactic connection 
does carry over also in the quantum regime. It is a widely open question 
to give deep meaning to the connection between 
$\qLTC$s and quantum local testability 
of proofs, as is known in the classical case \cite{Gold}. 

\subsection{Overview of Proofs of Theorems \ref{thm:QECC} and 
\ref{thm:sound}}\label{subsec:overview}

\subsubsection{Bounds on $\sLTC$ codes on Expanders}
To prove theorem \ref{thm:QECC}, we want to use good small-set expansion 
in order to construct an error which will not have a large energy penalty (namely, will not violate too many constraints) 
but which will be of large weight. More precisely, the error should have  
a large weight modulo 
the centralizer of the stabilizer group (see Definition \ref{def:sound}), 
and yet should not
violate too many stabilizer generators 
(recall that an error violates a stabilizer 
generator, or constraint, if it does not commute with it;
see definition \ref{def:stab}).

The key idea is that in a small-set expander, 
intersections between stabilizer generators which consist of more than 
one qudit are rare (See fact \ref{fact:deg}).   
The size of the intersection matters since for 
two generators that intersect on 
a single qubit, the restrictions of those operators to that qubit
must {\it commute}, because the two generators commute overall
(see definition \ref{def:stab}).  
We note that it cannot be that {\it all} generators when restricted 
to a given qudit commute, 
because this would mean 
this qubit is trivial for the code (see remark at the end of 
Subsection \ref{sec:QECCdefs}). 
An error defined on a qudit in such a way that it 
commutes with the majority of the 
generators acting on it, will violate only a small  
fraction of the 
constraints acting on that qudit. 

To extend this to errors of larger weight (up to some small 
constant fraction), we apply the above idea to each of 
the generators in  
a large ``sparse'' set of generators, namely 
a set in which each two terms are of at least 
some constant distance apart in the interaction bi-partite graph.  
(formally, a $1$-independent set of terms;  
see Definition \ref{def:Lind}).  
It is not difficult to see that due to the distance between the
generators, the error weight remains large 
even modulo the centralizer. 

\subsubsection{Upper bound on soundness for stabilizer $\sLTC$s on any 
graph}  
To prove theorem \ref{thm:sound}, we want to prove that 
regardless of the  graph they are set on, the  
relative soundness of $\qLTC$s is bounded from above by some constant 
strictly smaller than $1$. 
We use the bound of theorem (\ref{thm:QECC}) 
(the "surprising" side) augmented 
with a claim that quantum stabilizer codes not only 
suffer from the
quantum effect of Theorem (\ref{thm:QECC}) 
but also cannot avoid the classical  
effect by which codes with {\it poor} small-set expansion have low soundness,
namely that large error patterns are examined by relatively few
check terms, so the number of constraints they violate is relatively low.
Together, this means that for {\it any} underlying graph, 
whether a good or a bad small set expander, 
the relative soundness is non-trivially bounded. 

While in the classical case, the fact that poor expansion implies poor 
relative soundness, is very easy to argue,  
in the quantum case the proof turns out to be 
quite non-trivial, but still a similar phenomenon holds. 
Let us clarify what we're trying to show.
We want to show that if the expansion is bad, one can construct 
an error of large weight but which does not 
have large relative penalty.
Suppose we would like to show that the soundness function $r(\delta)$ 
is small, for some range of proximity values $(0,\delta_0]$. 
Consider a set of qudits $S$ whose fractional size is some $\delta \in (0,\delta_0]$, and which has positive expansion error 
$\eps>0$. A priori, 
if we have an error supported on $S$, then the maximal number of violations 
is at most $|S| D_L (1-\eps)$, 
by the assumption on the expansion.
This might seem as though it proves the result trivially.
The technical problem here, however, is that 
an error on $S$ may just "seem" to be large, whereas possibly, may be 
represented much more succinctly modulo the centralizer group. 
This problem is, once again, inherently quantum - it corresponds, essentially, 
to showing that a given error has large weight even modulo to 
the {\it dual} code, 
namely the code spanned by the generators themselves.
We would hence like to devise an error pattern, that 
cannot be downsized significantly by operations in the centralizer group, 
but would still "sense"
the non-expanding nature of $S$, and hence have fewer-than-optimal violations.

To this end we prove the Onion fact (Fact \ref{fact:succinct}) which 
might be of interest of its own. It  
states that given an error on at most $k/2$ 
of the $k$ qudits supporting  
a generator, its weight cannot be reduced modulo the centralizer 
within the $k$-neighborhood 
of the
generator (the $k$ neighborhood is, roughly,  
the qudits belonging to the set of 
terms of distance $k$ from that generator 
in the interaction graph). 
The ``Onion'' in the name is due to the 
fact that the proof (given in Subsection \ref{sec:onion})
works via some hybrid argument on
the onion-like layers $\Gamma^{(i)}(u)$ 
surrounding the qudits of a generator $u$. 

Our idea is to concentrate the error on 
a large set of far away generators whose $k$-neighborhoods 
are non-intersecting (we call those generators ``islands''). 
We now argue as follows. If we draw a random error on the qudits belonging to 
these "islands", with probability calibrated so
that the expected number of errors per "island", 
is, say, $1$ error, the following will occur: 
on one hand, many islands have more than one error, so they ``sense'' 
the sub-optimality of expansion.  
On the other hand, only a meager fraction, exponentially small in $k$, 
of the "islands" with at least two errors, will have more than $k/2$ errors; 
only those, by the Onion fact (fact \ref{fact:succinct}) 
can be potentially reduced modulo the centralizer. 
Hence with high probability,
the weight of the random error, cannot be significantly reduced
modulo the centralizer, yet it still has 
less-than-optimal number of violations due to the expansion. 

\subsection{Related work}
Theorem \ref{thm:QECC} 
is related to our recent result \cite{CLH}
in which it was shown 
that when a quantum local Hamiltonian, whose terms mutually 
commute, is set on 
a good small-set expander, then the approximation of its ground energy 
lies in $\NP$. In that result, 
the better the small-set expansion, the better the 
approximation. In other words, as the expansion improves, the problem becomes 
less interesting from the 
quantum point of view. Another result of the same spirit
was derived by Brandao and Harrow \cite{BH} for non-commuting $2$-local 
Hamiltonians on standard expanders.   
In both results good expansion poses a limitation on the  
expressiveness of quantum constraint systems. 
We note that the starting point of both the proof of our
Theorem \ref{thm:QECC} and the result of \cite{CLH} are 
Facts \ref{fact:essence} and \ref{fact:deg} regarding the percentage 
of unique neighbors in good small set expanders; however, 
the proofs proceed from that point onwards in very different directions. 

Dinur and Kaufman \cite{Din3} 
showed that classical $\LTC$ codes {\it must} be set on 
a good small-set expander. 
More precisely, given a code with soundness 
$R(\delta) = \rho \cdot \delta$ for all 
$\delta>\delta_0$ for some constant $\delta_0$, 
the edge expansion of the underlying graph 
is at least $c\rho$, for some constant $c$.
This might seem to provide another classical contrast to our Theorem 
\ref{thm:QECC}, in addition to our Claim \ref{cl:classical}.  
However, 
\cite{Din3} does not use bi-partite graph expansion but rather the 
graph in which an edge connects any two nodes that participate in a common 
constraint; the two notions of expansion are very different 
and hence direct comparison to the \cite{Din3} result 
is not possible.   

\subsection{Discussion and Further directions}\label{sec:discussion}
Many open questions arise regarding $\qLTC$s. 
Can we find other $\qLTC$s with much better parameters than 
those mentioned in this article?  
It is a natural starting point to check known quantum codes that have good 
self-correcting properties, or high 
energy barrier \cite{Haah,M}.
Do $\qLTC$s
exist with parameters which are as good as those of \cite{Din, BS}, 
namely, constant distance, constant query complexity, constant soundness for 
all proximities larger than some constant $\delta_0>0$, 
and rate which is inverse 
polylogarithmic? 
If not, can we prove appropriate upper bounds on $\qLTC$s?  

The upper bounds we provided here 
point to an inherently quantum phenomenon, 
which constitutes an obstacle against local testability for 
$\qLTC$s in the low-proximity range of parameters. 
Both of our main theorems, reflect, in fact, a deeper phenomenon called {\it monogamy of entanglement}
which was identified also in \cite{CLH} for commuting local Hamiltonians, and \cite{BH} for $2$-local general Hamiltonians.
Essentially, this phenomenon limits the amount of entanglement that a single qudit with $O(1)$ quantum levels can "handle".
In quantum codes, based on commuting check terms, the entanglement of code states arises 
from the fact that the operators actually do not commute per qubit, but only 
over sets of qubits.
Incidentally, per-qubit non-commutativity is also the phenomenon 
responsible for the energy "penalty" received by certain (sparse) errors.
Hence, in cases where monogamy of entanglement is a significant factor, for example in small-set expander geometry,
we witness an inherent decline in the energy "penalty" of such errors, thus upper-bounding the quantum local testability.
It is thus the combination of {\it monogamy of entanglement} in small-set expanders, and the poor local testability of non-expanders,
that are responsible for the apparently quantum phenomenon.
Whether Theorem \ref{thm:sound} 
hints at a more profound limitation on quantum
local testability, that holds also  
for larger values of $\delta$, 
calls for further research. 
Perhaps refuting the $c^3$ open problem is doable 
in the quantum case? 

Finally, the link between quantum local testability of proofs and 
$\qLTC$s, so crucial in the classical world \cite{Gold},
is far from clear in the quantum setting. 
We have merely touched upon it (see the 
result of quantum $\PCPP$s in the 
appendix), however, much further clarification of this connection, 
is called for. 

{~}

\noindent
\textbf{Organization of paper} 
In Section \ref{sec:bg} we provide the necessary background 
on quantum error correcting codes and on small-set expanders. 
Section \ref{sec:qLTC} provides definitions of quantum locally 
testable codes ($\qLTC$s) and stabilizer $\qLTC$s , and basic results. 
Section \ref{sec:QECC} provides bounds on the soundness of 
quantum $\LTC$s on small-set expanders, and 
Section \ref{sec:sound} provides an absolute bound on soundness 
of stabilizer $\LTC$s regardless of the expansion of their 
underlying graph. 
Finally, In the Appendices we 
provide several proofs which are on the more technical side. 
In Appendix \ref{app:PCPP} we provide  
our definition of quantum $\PCPP$s and 
the construction and proof of the induced $\qLTC$. 

\section{Background}\label{sec:bg}

\subsection{The Pauli groups}

\begin{definition}

\textbf{Pauli Group}

\noindent
The group $\Pi^n$ is the $n$-fold tensor product of Pauli operators $A_1\otimes A_2 \otimes \hdots \otimes A_n$, where $A_i\in \left\{I,X,Y,Z\right\}$,
along with multiplicative factors $\pm 1, \pm i$ with matrix multiplication as group operation.
\end{definition}
The Pauli group can be generalized to particles of any dimensionality
$d$: 
\begin{definition}\label{def:generalpauli}
\textbf{The Pauli group generalized to $F_d$ }

\noindent
Let $X^k_d:|i\rangle \mapsto |(i+k) \pmod d\rangle , 
P_d^{\ell}|j\rangle \mapsto w_d^{j\ell}|j\rangle$ 
be the generalized bit and phase flip operators on the 
$d$-dimensional Hilbert space, where $w_d=e^{2\pi i/d}$ 
is the primitive $d$-th root of unity.
Let $\Pi_d$ be the group generated by these operators and all roots of 
unity of order $d$. 
The group $\Pi_d^n$ is the $n$-fold tensor product of Pauli operators $A_1\otimes A_2 \otimes \hdots \otimes A_n$, where $A_i\in \left\{X_d^kP_d^\ell\right\}$
along with these multiplicative factors.  
\end{definition} 
The weight of a Pauli operator is defined to be the 
number of locations where it is non-identity. 

\subsection{General Quantum Error Correction}

\begin{definition} 

\textbf{Quantum Code}\label{def:code}

\noindent A quantum code on $n$ qudits is given by 
a set of ($m$) projections $\Pi_i$. The code is defined to be 
the simultaneous $0$ eigenstates of all those projections. 
\end{definition} 

\begin{definition}

\textbf{Quantum Error detection 1}\label{def:det1}\cite{Knill}

\noindent
Let $C\subseteq {\cal H}$ be a quantum code on $n$ qudits. 
Let $\Pi_C$ be the orthogonal projection onto $C$.
We say that the set of errors ${\cal E}$ 
is detectable by $C$ if for any $E\in {\cal E}$, we have:
\begin{equation}\label{eq:genqecc}
\Pi_C E \Pi_C = \gamma_{E} \Pi_C,
\end{equation}
where $\gamma_{E}$ is some constant which may depend on $E$. 
\end{definition}

\begin{definition}

\textbf{Quantum Error detection 2}\label{def:det2}\cite{Knill}

\noindent
A set ${\cal E}$ is detectable by $C$, if 
for any $\ket{\psi},\ket{\phi}\in C$ with 
$\langle \psi | \phi \rangle=0$, and any $E\in {\cal E}$,
$\bra{\psi} E \ket{\phi}=0$.
\end{definition}

\begin{claim}\label{cl:defequiv}\cite{Knill}
Definitions (\ref{def:det2}) and (\ref{def:det1}) are equivalent:
\end{claim}

The proof can be found in the Appendix. 
Definition (\ref{def:det2}) gives rise to the following
 natural definition:
\begin{definition}\label{def:qeccdist}

\textbf{Distance of a code}\cite{Knill}

\noindent
Let $C$ be a quantum code detecting error set ${\cal E}\subset \Pi_d^n$. 
$C$ has distance $dist(C)$ if for any two orthogonal code states $\ket{\phi},\ket{\psi}$, and any $E\in {\cal E}$ of weight at most $dist(C)-1$, we have $\langle \phi | E|\psi \rangle = 0$.
\end{definition}

\subsection{Stabilizer Quantum Error Correcting Codes}\label{sec:QECCdefs}

\begin{definition}\label{def:stab} 
\textbf{Stabilizer Code}

\noindent
A stabilizer code $C$ 
is defined by an Abelian subgroup $A=A({\cal G})\subset \Pi_d^n$, generated by a set ${\cal G} \subset \Pi_d^n$.
The codespace is 
defined as the mutual $1$-eigenspace of all elements in ${\cal G}$ 
(we require that $-I\notin {\cal G}$ so that this codespace is 
not empty). 
An element $E \in \Pi_d^n$ is said to be an error if 
it does not commute with at least one element of ${\cal G}$, 
i.e. $E \notin \mathbf{Z}({\cal G})$, where 
$\mathbf{\mathbf{Z}}({\cal G})$ is the centralizer of ${\cal G}$.
An element $E \in \Pi_d^n$ is said to be a logical operation, if 
it commutes with all of ${\cal G}$, but is not generated by ${\cal G}$, i.e., 
$E \in \mathbf{Z}({\cal G})-A.$
A stabilizer code is said to be $k$-local if each term $g\in {\cal G}$
is an element of $\Pi_d^n$, with weight exactly $k$.
\end{definition}
To fit with the terminology of Definition \ref{def:code}, consider for each 
generator $g$ the 
projection $\Pi_g$ which projects on the orthogonal subspace to the $1$ 
eigenspace of $g$.

\begin{definition}
\textbf{Succinct representation}

\noindent
A $k$-local set of generators ${\cal G}$ is said to be {\it succinct}, if there does not exist
a different generating set 
${\cal G}'$, such that $A({\cal G}) = A({\cal G}')$ and $wt(g)<k$ for some $g\in {\cal G}'$.
\end{definition}

\noindent
The following is a well known fact \cite{Got}
which will be useful later on, and we prove it in appendix 
(\ref{sec:lemPauli}).
\begin{lemma}\label{lem:Pauli}
\textbf{Stabilizer Decomposition}

\noindent
Let $C$ be a stabilizer code on $n$ qudits, and consider the sets 
$EC= \left\{E \ket{\phi} , \ket{\phi}\in C\right\}$ 
with $E\in \Pi_d^n$.
Then two sets $EC$, $E'C$ are either orthogonal or equal to each other, 
and $\left\{EC\right\}_{E\in \Pi_d^n}$ span the entire Hilbert space. 
Moreover, consider the partition of the entire Hilbert space to 
sets of states which are mutual eigenvectors of all generators of $C$ with 
exactly the same set of eigenvalues for each generator. 
Then this partition is exactly the partition derived by the $EC$'s, 
and two orthogonal $EC$'s have two lists of eigenvalues which differ 
on at least one generator. 
In particular, any $n$ qudit state $\ket{\psi}$ may be written as
a sum of orthogonal vectors $$ \ket{\psi} = \sum_i E_i \ket{\eta_i},$$
where $E_i\in \Pi_d^n$ and $\ket{\eta_i}\in C$.
\end{lemma}

%

\begin{definition}\label{def:weight}
\textbf{Weight of an error in stabilizer codes}

\noindent
Let $C$ be a stabilizer code on $n$ $d$-dimensional qudits, with 
generating set ${\cal G} \subset \Pi_d^n$. 
For $E\in \Pi_d^n$, we denote:
\begin{enumerate}
\item
The number of locations in which 
$E$ is non-identity - by $wt(E)$.
\item
The weight of $E$ modulo the group 
$A({\cal G})$ - by $wt_{\cal G}(E)$:
$wt_{\cal G}(E)=\min_{f\in A({\cal G})}\{wt(fE)\}$.  
\item
The weight of $E$ modulo the centralizer 
$\mathbf{Z}({\cal G})$ - by $wt_{\mathbf{Z}({\cal G})}(E)$:
$wt_{\mathbf{Z}({\cal G})}(E)=\min_{z\in \mathbf{Z}({\cal G})}\{wt(zE)\}$.  
\end{enumerate}
\end{definition}

\noindent
The above claims give rise to the following definition of distance in a stabilizer code:
\begin{definition}\label{def:stabdist}
\textbf{Distance of a stabilizer code}

\noindent
Let $C$ be a $k$-local stabilizer code on $n$ $d$-dimensional qudits, 
with generating set ${\cal G}\subset \Pi_d^n$.
The distance of $C$ is defined as the minimal weight of any logical operation on $C$:
$$dist(C) = min_{E\in \mathbf{Z}({\cal G})-A({\cal G})} wt(E).$$
\end{definition}

\begin{claim}\label{cl:distequiv} {\bf Equivalence of distance definitions}
A stabilizer code $C$ has $dist(C)\geq \rho$ by definition \ref{def:stabdist}, 
iff it has distance $\geq\rho$ by definition \ref{def:qeccdist}. 
\end{claim} 

The proof is given in the appendix subsection (\ref{sec:distequiv}).
A code $C$ on $n$ qudits is said to have a constant relative distance $\delta>0$, if its distance
is at least $\delta n$. 
We will make use of the following assumption which we isolate so that 
we can refer to it later on:

{~}

\noindent{\bf Remark:}
If there is a qudit $q$ 
such that 
all states in the code 
look like $|\alpha\rangle$ tensor with some state on the remaining qudits, 
for some fixed one-qudit state $|\alpha\rangle$ of that qudit $q$, 
we say that $q$ is {\it trivial} for the code. 
We will assume in the remainder of the paper that for all codes we handle, 
no qudits are trivial for the code,
since such qudits can be simply discarded.  

\subsection{Interaction graphs and their expansion} 
We assume in the rest of the paper that each qudit
participates in exactly $D_L$ constraints. 
We define bi-partite expanders, similar to 
\cite{Spi}, \cite{CRVW}, who used them to 
construct locally-testable classical codes. 
Note that we require expansion to hold only for sets of constant size $k$.

\begin{definition}\label{def:bipgraph}

\textbf{Bi-Partite Interaction Graph}

\noindent
Let $C$ be a quantum code on $n$ $d$-dimensional 
qudits, whose check terms $\left\{\Pi_i\right\}_i$  
are $k$-local.
We define the bi-partite interaction graph of $C$ $G=G(C) = (L,R;E)$ as follows: 
the nodes $L$ correspond to the qudits, the nodes $R$ correspond to the check terms,
and the set of edges connect each constraint $\Pi_i\in R$ 
to all the qudits in $L$ on which it acts non-trivially. 
We note that $G$ is left $D_L$-regular, and right $k$-regular.
\end{definition}

\begin{definition}\label{def:expbi}

\textbf{Bi-partite expansion}

\noindent
Let $G=(L,R;E)$ be a bi-partite graph, that is left $D_L$-regular, 
right $k$-regular.
A subset of qudits $S\subseteq L$ is said to be $\eps$-expanding, if $|\Gamma(S)| \geq |S| D_L (1-\eps)$, where $\Gamma(S)$ is the 
set of neighbors of $S$ in this graph. $\epsilon$ is called the expansion error 
for this set. 
$G$ is said to be $\epsilon$-small-set-expanding, if every subset 
$S\subseteq L$, $|S|\leq k$ has expansion error at most $\eps$.
\end{definition}

We state two technical facts on good bi-partite expanders that will be useful later on.
The proofs are in the appendix (\ref{sec:bipartite}).
\begin{fact} \label{fact:essence}
Consider $S\subseteq L$  in a bi-partite graph $G(L,R:E)$ 
and let $S$ be $\epsilon$-expanding, for $\eps<\frac{1}{2}$.
Then a fraction at most $2\eps$ of all
vertices of $\Gamma(S)$ have degree strictly larger than $1$ in $S$.
\end{fact} 

\begin{fact}\label{fact:deg}
Let $S\subseteq L$ in a 
bi-partite graph $G = (L,R;E)$, such that $S$ is $\eps$ expanding, for $\eps<\frac{1}{2}$.
Then there exists a vertex $q\in S$, 
such that the fraction of neighbors of $q$ with at least two neighbors in $S$
is at most $2\eps$.
\end{fact}

\subsection{Notation}
We denote as follows. 
$d$ is the dimension of the qudits involved.
For a bi-partite graph we denote $G=(L,R;E)$, $L$ denotes the 
left set of vertices of size $|L|=n$ (corresponding to qudits), 
$R$ denotes the right vertices $|R| = m$ (corresponding to constraints), 
and $E$ is the set of edges between $L$ and $R$.  
$D_L$ will denote the left degree of a bi-partite graph. 
$k$ will denote the locality of the 
constraints, namely the right degree of the graph. 
Given $S\subseteq R$ (or $L$) in a bi-partite graph,   
$\Gamma(S)$ denotes the neighbor set of $S$ in $L$ (or $R$). 
${\cal N}(q)$ 
will denote the qudit-neighborhood 
of a qudit $q$ in $L$, namely 
all the qudits participating in all the constraints acting on $q$ 
(so, ${\cal N}_q = \Gamma^{(2)}(q)$). 
We will use
$\epsilon$ 
to denote the expansion error for bi-partite 
graphs (as in Definition \ref{def:expbi}). 
We will use $\delta$ (and sometimes $\mu$)
to denote the proximity, namely, the 
relative distance of a word from a code.

\section{Locally-testable quantum codes}\label{sec:qLTC}
In this section we define locally testable quantum codes, 
both in the general case, and in the specific case of stabilizer codes.
We then show that our definitions coincide for stabilizer codes. 

\subsection{Local testability of general quantum codes}

We first generalize definition (\ref{def:qeccdist}),
from a definition of distance
{\it of a code} to a definition of distance {\it from a code}: 
\begin{definition}\label{def:distcode}

\textbf{Distance from a quantum code}

\noindent
Let $C$ be a quantum code detecting error set ${\cal E}\subset \Pi_d^n$.
For any two orthogonal states $\ket{\phi},\ket{\psi}\in {\cal H}$,  
we define the Hamming distance 
between them 
$dist_C(\ket{\phi},\ket{\psi})$
as the maximal integer $\rho$, 
such that for any $E\in {\cal E}$, 
\dnote{why restrict to this and not 
any Pauli of this weight?}
\lnote{We could of course, but this definition is more general, because it treats general errors.  I think it should be kept that way.}
\dnote{this is very confusing, what's the advantage? let's discuss again}
with $wt(E)\leq \rho-1$, 
we have $\bra{\psi} E \ket{\phi}=0$.
Similarly, given a state $\ket{\phi}$ orthogonal to $C$, 
we say that the distance of
$\ket{\phi}$ from $C$ denoted by $dist(\ket{\phi},C)$ is the minimum 
over all $\ket{\psi}\in C$ of  $dist_C(\ket{\phi},\ket{\psi})$. 
\end{definition}

We note here that the distance of a state {\it from the code} in the above, can be much larger than the distance {\it of the code}.
This, akin to the classical case, where locally-testable codes are required 
to identify words far from the code, even if they cannot
be (uniquely) 
decoded, so that these codes can be used as proof systems.

\noindent
\begin{definition}

\textbf{Quantum Locally Testable Codes ($\qLTC$)}\label{def:qLTC}

\noindent
Let $R = R(\delta)$ be some function $R(\delta): [0,1] \mapsto [0,1]$,
this is called the soundness function.
Let $C$ be a quantum code  
on $n$ $d$-dimensional qudits, defined as the groundspace of
$H=\sum_{i=1}^m\Pi_C^i$, where 
$\Pi_C^i$ are $m$ $k$-local projections for some constant $k$. 
We say that $C$   
is {\it quantum locally testable} with soundness $R(\delta)$,
if:
$$
\forall \delta>0, \ket{\Psi}: ~~
\mbox{   }
dist(\ket{\Psi},C) \geq \delta n
\mapsto
\frac{1}{m}\langle
\Psi |H|\Psi
\rangle 
\geq
R(\delta).
$$
The query complexity of the code is defined to be $k$. 
\end{definition}

\subsection{Local testability of quantum stabilizer codes}

We now show that local testability defined above (Definition \ref{def:qLTC}) 
has a natural interpretation in the context of stabilizer codes.

\begin{definition}\label{def:sound}
\textbf{Local Testability for Stabilizer Codes ($\sLTC$)}\label{def:sLTC}

\noindent
Let $R(\delta)$ be some function $R(\delta): [0,1] \mapsto [0,1]$.
We say that a stabilizer code $C$ on $n$  $d$-dimensional qudits is an  
$\sLTC$ with query complexity $k$ and soundness $R(\delta)$,
if there exists a generating set ${\cal G}$ for $C$,
where each element has support  
$k$, such that the following holds: 
for any $E\in \Pi_d^n$ with $wt_{\mathbf{Z}({\cal G})}({\cal E}) \geq \delta n$, 
a uniformly random generator $g\in {\cal G}$ does not commute with $E$ w.p. at least $R(\delta)$.
\end{definition}

\subsubsection{Equivalence of definitions of locally testable codes}

We now show that the definition of stabilizer locally testable codes 
(Definition \ref{def:sLTC}) is in fact a special case of the general quantum locally testable codes (Definition \ref{def:qLTC}). 

\begin{claim}\label{cl:sLTCqLTC}

\noindent
\begin{enumerate}
\item
If $C$ is a Stabilizer code with generating set ${\cal G}$, 
which is an $\sLTC$ 
with query complexity $k$, and soundness $R(\delta)$,
then the set of projections $\left\{\Pi_g\right\}_{g\in {\cal G}}$, where $I-\Pi_g$ is the projection on the $1$-eigenspace 
of $g$, defines a $\qLTC$ with query complexity $k$, and soundness $R(\delta)$.
\item
If $C$ is a $\qLTC$ 
with query complexity $k$, and soundness $R(\delta)$,
defined by a set of projections 
$\left\{\Pi_g\right\}_{g\in {\cal G}}$,
such that the set $\left\{I-\Pi_g\right\}_{g\in {\cal G}}$ spans an Abelian subgroup of $\Pi_d^n$,
then $C$ is also an $\sLTC$ 
with query complexity $k$, and soundness $R(\delta)$. 
\end{enumerate}
\end{claim}

\begin{proof}

\noindent
\textbf{$\sLTC \mapsto \qLTC$}

\noindent
By definition of a stabilizer code, for any $\ket{\phi}\in C$, we have $g\ket{\phi}=\ket{\phi}$ for all $g\in {\cal G}$,
so $\Pi_g \ket{\phi}=0$ for all $g\in {\cal G}$.
Next, consider a state $\ket{\phi}$ orthogonal to $C$, such that $dist(\ket{\phi},C)\geq \delta n$.
We would now like to show that a projection chosen randomly 
from $\left\{\Pi_g\right\}_{g\in G}$
is violated by $\ket{\phi}$ with 
probability at least 
$R(\delta)$. 
Consider the following orthogonal decomposition of $\phi$ as 
implied by lemma (\ref{lem:Pauli}):
\begin{equation}\label{eq:partition}
 \ket{\phi} = 
\sum_i \alpha_i \ket{\alpha_i} = \sum_i \alpha_i E_i \ket{\eta_i},
\end{equation} 
where $E_i\in \Pi_d^n$, $\ket{\eta_i}\in C$, and 
$E_i \ket{\eta_i}$ are orthogonal. 
We claim that for each $i$, $wt_{\mathbf{Z}({\cal G})}(E_i) \geq \delta n$:
otherwise, it is easy to see that 
there exists some $E'\in \Pi_d^n$, $wt(E')<\delta n$, such that for 
at least one $i$, 
we have $E' E_{i} \in \mathbf{Z}({\cal G})$. 
Since for any $J\in\mathbf{Z}({\cal G})$, $JC=C$, we have  
that alternatively, $E' \ket{\alpha_i} \in C$.
Since $E'$ is unitary, and the $\ket{\alpha_i}$'s are orthogonal, then the $E' \ket{\alpha_i}$'s are orthogonal, thus $E' \ket{\phi}$ has a non-zero projection on $C$.
Contrary to the assumption that $dist(\ket{\phi},C)\geq \delta n$.

If $E_i$ and $g\in {\cal G}$ do not commute, 
$E_i g = \omega g  E_i$, for some $\omega \neq 1$. 
In particular, $E_i\ket{\eta_i}$ 
is a $\omega$ eigenstate of $g$. This means it is orthogonal to 
the $1$-eigenspace of $g$, and therefore:
$$
\bra{\alpha_i} \Pi_g \ket{\alpha_i} = 1.
$$
Yet, by the $\sLTC$ property of $C$, for each $i$, $E_i$ does not commute with a fraction at least $R(\delta)$ of the generators of ${\cal G}$. 
Thus, a randomly chosen check term is violated by $\ket{\alpha_i}$ with probability at least $R(\delta)$,
so
$$
\frac{1}{|{\cal G}|}
\sum_{g\in {\cal G}} \langle \alpha_i | \Pi_g | \alpha_i \rangle
\geq
R(\delta).
$$
Since by lemma (\ref{lem:Pauli}) the decomposition above coincides with the simultaneous eigenbasis of ${\cal G}$,
we have:
$$
\frac{1}{|{\cal G}|}
\langle
\phi | 
\sum_{g\in {\cal G}} \Pi_g |
\phi
\rangle
=
\frac{1}{|{\cal G}|}
\sum_i \sum_{g\in {\cal G}} |\alpha_i|^2 \langle \alpha_i | \Pi_g | \alpha_i \rangle
\geq
R(\delta).
$$

\noindent
\textbf{$\qLTC \mapsto \sLTC$}

\noindent
First, by definition, the set of states that are in the mutual 
groundspace of the $\Pi_g$'s, are stabilized (i.e. eigenvalue $1$) w.r.t. 
the terms ${\cal G}$, and vice versa. 
Now, let $E\in \Pi_d^n$, whose weight modulo $\mathbf{Z}({\cal G})$ is 
at least $\delta n$.
Let $\ket{\phi}\in C$ be any code state, and denote $\ket{\psi} = E \ket{\phi}$.
We claim that $dist(\ket{\psi},C)\geq \delta n$.
Otherwise there exists $E'\in \Pi^n$, $wt(E')<\delta n$,
such that $E' \ket{\psi}$ has a non-zero projection on $C$, 
hence $E' E \ket{\phi}$ has a nonzero projection on $C$,
so by lemma (\ref{lem:Pauli}), we have that $E'E C = C$.
Therefore, $E' E$ commutes with all ${\cal G}$, 
and hence $E' E\in \mathbf{Z}({\cal G})$, which implies that
$wt_{\mathbf{Z}({\cal G})}(E) < \delta n$, in contradiction.
By the $\qLTC$ property of $C$, we have
\begin{equation}\label{eq:gpenalty}
\langle
\psi |
\sum_{g\in {\cal G}} \Pi_g |
\psi
\rangle
\geq
|{\cal G}| \cdot R(\delta).
\end{equation}
Since $\ket{\psi} = E \ket{\phi}$, then for any generator $g$
$g \ket{\psi} = g E \ket{\phi} = \omega E g \ket{\phi} = \omega E \ket{\phi}$, for some $\omega\in \mathbf{C}$.
So for any $g\in {\cal G}$ $\ket{\psi}$ is some eigenstate of $g$.
Hence $\ket{\psi}$ is either 
in the $1$-eigenspace of $\Pi_g$
or in its $0$-eigenspace, so by equation (\ref{eq:gpenalty}) it violates
a fraction at least $R(\delta)$ of all generators ${\cal G}$.
\end{proof}

\section{Bound on the soundness of stabilizer $\LTC$s on small-set 
expanders}\label{sec:QECC}
In this section we prove theorem \ref{thm:QECC}. 
We define the {\it relative soundness} formally: 
\begin{definition}\label{def:smallr}{\bf Relative Soundness}
Define 
$$
r(\delta) : [0,1] \mapsto [0,1],
$$
as follows: 
$r(\delta)=R(\delta)/\Theta(\delta)$, where 
$\Theta(\delta)\equiv \min\{\delta k,1\}$. 
\end{definition}
We note that in the all the following, we will be interested 
in $\delta<1/k$ and in this range $r(\delta)=R(\delta)/k\delta$. 

\subsection{A useful fact about restrictions of stabilizers} 

\begin{definition}
\textbf{Restriction of stabilizers}

\noindent
For a $E\in \Pi_d^n$, let $E|_q$ denote the $q$-th component of the tensor product $E$, and let $E|_{-q}$ denote the tensor product of all terms except the $q$-th.
Similarly, for a generating set ${\cal G}$, we denote by ${\cal G}|_q$ as the set $\left\{g|_q \mbox{, } g\in {\cal G} \right\}$, and similarly for ${\cal G}|_{-q}$.

\end{definition}

\noindent
We now prove a useful fact: that the restrictions to a given qudit 
$q$ of all the generators of a stabilizer code with absolute 
distance strictly larger than 
$1$ cannot all commute.  

\begin{fact}\label{fact:nocommute}
Let $C$ be a stabilizer code 
with absolute minimal distance strictly larger than $1$.
Then for any qudit $q$, and any generator $g$ acting on $q$, there exists 
another
generator $h(q)$ 
acting on $q$ such that $[g|_q, h|_q] \neq 0$.
\end{fact}

\begin{proof}
Assume on the negative, that there is a qudit $q$,
and a generator $g$, such that for all other generators $h$, we have $[g|_q,h|_q]=0$.
Let $Q=g|_q$.
We have that $Q'=Q\otimes I_{-q}$, namely the tensor product with identity 
on the other qubits, commutes with all $g\in {\cal G}$, and thus 
$Q'\in \mathbf{Z}({\cal G})$. 
However, $Q'$ cannot be inside $A({\cal G})$, since otherwise 
$q$ is in 
some constant state (the $1$ eigenvector of $Q$) $\ket{\alpha}$
for all code states, and thus $q$ is trivial for the 
code (see remark at the end of Subsection \ref{sec:QECCdefs}). 
Hence, $Q'\in \mathbf{Z}({\cal G})-A({\cal G})$, so the distance of the code by definition (\ref{def:stabdist}) is $1$, in contradiction to our assumption. 
\end{proof}

\subsection{Proof of Theorem \ref{thm:QECC}}
In the proof we will make use of "sparse'' sets of constraints, 
defined as follows. 

\begin{definition}
\textbf{$1$-independent set of constraints}\label{def:Lind}

\noindent 
For a given constraint $u$, consider $\Gamma^3(u)$, the set of qudits 
acted upon by constraints which act on qudits in $u$. 
A set of constraints $U$ is said to be $1$-independent if for any 
two constraints $u,w \in U$, 
$\Gamma^3(u)\cap\Gamma^3(w)=\Phi$.  
\end{definition}

\begin{proof}(Of theorem \ref{thm:QECC}) 

\paragraph{Generating the error}
We want to construct an error $E \in \Pi_d^n$, $wt_{\mathbf{Z}({\cal G})}(E)\geq \delta n$, that will not violate too many constraints in ${\cal G}$.  
Let $C$ be a stabilizer code 
with a $k$-local generating set ${\cal G}$, such that the
bi-partite interaction graph of $C$
is an $\eps$ small-set bi-partite expander. 
Let $U$ be a $1$-independent set of constraints of size $\delta n$.
We note that since $\delta \leq \frac{1}{k^3 D_L}$ 
a $1$-independent set of this size must exist, by a simple greedy algorithm. 
For a given constraint 
$u\in U$, and $i\in [k]$, let $\alpha_i(u)$ 
denote the number of 
generators $g\in {\cal G}$
that act on a qudit $i$ in $u$ and intersect $u$ in at least one other
qudit.
Then for each $u\in U$ we define $q(u)$ to 
be a qudit of minimal $\alpha_i(u)$ over all $i\in [k]$.
Let $T = \left\{q(u) | u\in U\right\}$.
Let us define an error pattern:
$$ E = \bigotimes_{u\in U} u|_{q(u)}.$$

We first note that $E \notin \mathbf{Z}({\cal G})$; 
This is true by Fact (\ref{fact:nocommute}):
for each qudit $q$ in the support of $E$, 
$E|_q$ does not commute with $h|_q$ for some $h\in {\cal G}$.
But since $T$ is induced by a $1$-independent set, $h$ 
does not touch any other qudit in the support of $E$ except $q$, 
so this implies $[h,E]=[h|_q,E|_q]\neq 0$. 
We will now show that $E$ has large weight modulo $\mathbf{Z}({\cal G})$, 
but is penalized by a relatively small fraction of ${\cal G}$.

\paragraph{Weight Analysis}
By definition, we have that $wt(E) = |T|=|U|=\delta n$.
We claim that:
\begin{equation}\label{eq:tweight}
wt_{\mathbf{Z}({\cal G})}(E)= |T|
\end{equation}
Since  $\delta$ was chosen to be smaller than half the distance of 
the code $C$, 
$wt_{\mathbf{Z}({\cal G})}(E)= wt_{\cal G}(E)$
and so it suffices to lower-bound $wt_{{\cal G}}(E)$.
 
Suppose on the negative that $wt_{\cal G}(E)< |T|$.
Then there exists $\Delta\in A({\cal G})$, such that $E' = \Delta E$ has $wt(E')<|T|$.
Since the weight of $E'$ is strictly smaller than 
that of $E$, there must be one qudit $q_0$ in $T$, s.t. 
on the neighborhood ${\cal N}(q_0)$ the weight of $E'$ is strictly 
smaller than that of $E$, which is $1$; 
namely, $E'$ must be equal to the identity on all the qudits in the qudit-
neighborhood of $q_0$.
Here, we have used the fact that the qudit-neighborhoods of different qudits 
in $T$ 
are non-intersecting. This is true  
by the fact that the qudits were chosen by picking one qudit from each 
constraint out of  
a $1$-independent set of constraints 
(definition \ref{def:Lind}). 
This means that $\Delta$ must be equal to the inverse of $E$ 
on this neighborhood. But this inverse is exactly the following: 
It is equal to $E|_{q_0}^{-1}$ on $q_0$, and to the identity on all other 
qudits in the neighborhood.
By construction , $E|_{q_0}$ on $q_0$, (and therefore also 
$E^{-1}|_{q_0}=\Delta_{q_0}$) does not commute with $h|_{q_0}$, for some $h\in {\cal G}$.
Since $\Delta$ is identity on all qudits of $h$ other than $q_0$, this implies that 
$\Delta$ does not commute with $h$, in contradiction 
to the fact that $\Delta \in A({\cal G})$. 

\paragraph{soundness Analysis}
We upper-bound the number of generators that do not commute with $E$.
For each $u\in U$, the number of generators $g\in {\cal G}$ 
that do not commute with $E|_{q(u)}$ is at most
the number of generators that share at least two qudits with $u$.
By fact (\ref{fact:deg}) there exists a qudit 
$q\in \Gamma(u)$ such that the fraction of its check terms 
with at least two qudits in $\Gamma(u)$ is at most $2\eps$; 
since we chose $q(u)$ to be the qudit that minimizes that fraction 
over all qudits on which $u$ acts, 
we have that for $q(u)$, the fraction of terms acting on it 
that intersect $u$ with at least $2$ qudits is at most $2\epsilon$. 
Thus, the absolute number of generators acting on $q(u)$ 
that intersect $u$ in at least two qudits is at most $2\eps D_L$.
Hence the overall number of 
generators violated by $E$ is at most $2 \eps |T| D_L$. 
By Equation \ref{eq:tweight} this is equal to 
$2\eps D_L wt_{\mathbf{Z}({\cal G})}(E)$.
Using $D_Ln=mk$, 
we have $R(\delta)\le 2\eps k \delta$ and so $r(\delta) \leq 2 \eps$.
\end{proof}

We now show that a slightly stronger version of the above theorem holds.
This version will be used for showing Theorem (\ref{thm:sound}). 
\begin{claim} \label{cl:QECC}
Let $C$ be a good
stabilizer code, with a $k$-local succinct generating set, where each qubit is examined 
by $D_L$ constraints.
If there exists
a $1$-independent set of constraints $U\subseteq R$, 
s.t. $|U| = \delta n$ for some $0<\delta<1/k$, 
and $\Gamma(U)$,  
the set of qudits that the constraints in $U$ act on satisfies 
$|\Gamma(\Gamma(U))|\geq |\Gamma(U)| D_L (1-\eps)$, 
then for any $\delta'\leq \delta$ we have that 
$r(\delta')\le 2\epsilon$.
\end{claim}

\begin{proof}
For a set $S\subseteq L$, let $\Gamma_1(S)$ denote the number of neighbors of $S$ having a single neighbor in $S$, and let 
$\Gamma_{\geq 2}(S) \equiv \Gamma(S) - \Gamma_1(S)$.
Put $S = \Gamma(U)$, and let $S = \bigsqcup_{i=1}^k S_i$, denote a partition of $S$ into $k$
disjoint sets, where each $S_i$ takes a single (arbitray) qubit from each $\Gamma(u)$, $u\in U$.
By assumption, $|\Gamma(S)| \geq |S| D_L (1-\eps)$, whereas the total degree of $S$ is $|S| D_L$.
Hence, $|\Gamma_{\geq 2}(S)| \leq |S| D_L \eps$, so $|\Gamma_1(S)| \geq |S| D_L (1-2\eps)$.
Since each unique neighbor of $S$ examines exactly one partition $S_j$, 
there exists a partition $S_0$ examined by at least $|S_0| D_L (1-2\eps) = \delta n D_L (1-2\eps)$, constraints from $\Gamma_1(S)$.

Now, given any $\delta'\leq \delta$, 
let $S_0'$ be a subset of $S_0$ of size $\delta' n$, maximizing the ratio $\Gamma_1(S')/ |S'|$, over all sets $S'\subseteq S_0$ of this size.
Since each element of $\Gamma_1(S)$ examines just one element of $S$, such a set exists, with ratio at least $D_L (1-2\eps)$.
A tensor-product error ${\cal E}$ defined by taking, for each $u\in U$ the restriction to its qubit in $S_0'$,
we have by equation (\ref{eq:tweight}) $wt_{\mathbf{Z}({\cal G})}(E) = \delta' n$, 
whereas the maximal penalty is at most $2 \eps D_L \delta' n$.
Since $\delta'\leq \delta<1/k$ it follows that $r(\delta')\leq 2\eps$.
\end{proof} 

\section{An upper-bound on soundness}\label{sec:sound}

We now show an absolute constant strictly less than $1$, upper-bounding the 
relative soundness of any 
good quantum stabilizer code spanned by $k$-local generators, whose qudits 
are acted upon by $D_L$ stabilizers each.  
We start with an easy alphabet based upper bound.

\subsection{Alphabet-based bound on soundness}\label{sec:alphabet}

In attempting to understand soundness of good 
stabilizer codes, one must first account for
limitations on the soundness that seem almost trivial, 
and occur even when there is just a single error. 
\begin{definition}\label{def:single} 
{\bf Single error soundness}

\noindent
Let $t(d) = 1/(d^2-1)$; The single error
relative soundness in dimension $d$ 
is defined to be $\alpha(d)=1-t(d)$. 
\end{definition}

The motivation for the above definition is as follows. 
For any qudit $q$, there always exists $Q\in \Pi_d$, $Q\neq I$, 
such that a fraction at least $t(d)$ of the generators touching 
$q$ are equal to $Q$ when restricted to $q$. 
If we consider a single-qudit error on $q$ to be equal to 
$Q$, then it would commute with $t(d)$ of the generators acting on 
$q$; thus they can violate at most 
$\alpha(d)$ of the constraints acting on $q$.
Hence, one can expect that it is possible to construct an error of linear 
weight, whose relative soundness $r(\delta)$ 
is bounded by the single error relative soundness 
using qudits whose neighboring constraints are far from each other.

Indeed, we show:

\begin{fact}\label{fact:alphabet}
\textbf{Alphabet bound on soundness}

\noindent
For any good stabilizer code $C$ on $n$ $d$-dimensional qudits, 
with a $k$-local succinct generating set ${\cal G}$, whose left-degree is $D_L$,
we have $r(\delta) \leq \alpha(d)$, for any $\delta \leq 1/(k^3 D_L)$.
\end{fact}

\begin{proof}
Similarly to Theorem (\ref{thm:QECC}), given the parameters assumed 
in the statement of the fact, there exists a 
$1$-independent set of constraints $U$ of size $\delta n$.
For each constraint $u\in U$ we select arbitrarily 
some qubit $q=q(u)\in \Gamma(u)$ and examine 
the restrictions to $q$ of all stabilizers acting non-trivially on $q$.
Let $P(q)$ denote the set of all such restrictions. 
Let $MAJ(q)$ denote the element of $\Pi_d$ that appears 
a maximal number of times in $P(q)$.
We then set $E = \bigotimes_{u\in U} MAJ(q(u))$.
We first realize that $E$ is an error: 
we want to show that there exists a generator $g$ such that 
$E$ and $g$ do not commute.
Otherwise, 
$E$ commutes with all generators; 
Since by construction, each generator intersects $E$ with at most one qudit, 
this means that the restrictions to $q$ also commute: $[E|_q,g|_q]=0$
for all $q(u)$ acted upon by $E$. This is a  
contradiction by Fact (\ref{fact:nocommute}); hence, there must be a 
generator which does not commute with $E$, so $E$ is indeed an error. 
Similarly to the proof of Equation (\ref{eq:tweight}) in the proof of 
Theorem (\ref{thm:QECC}), we also have  
$wt_{\mathbf{Z}({\cal G})}(E) = \delta n$.
Furthermore, for each qudit $q$, the fraction of generators on $q$, whose restriction to $q$ does not commute with $E|_q$ is at most $\alpha(d)$, since the number
of appearances of $E|_q = MAJ(q)$ in $P(q)$ is at least $t(d) = 1-\alpha(d)$.
Hence the number of violated constraints is at most $\alpha(d) \cdot |U| 
\cdot D_L = \alpha(d)\delta n D_L$.
Since $\delta < 1/k$ it follows that $r(\delta) \leq \alpha(d)$.
\end{proof}

We note that classically, there is no direct analogue to the requirement of 
non-commutativity to achieve constraint violation. No analogue of the 
$\alpha(d)$ thus exists. 

\subsection{Separation from alphabet-based soundness}

In this section we show that the alphabet-based 
bound on the relative soundness 
in fact cannot be achieved, and the relative soundness is further bounded 
by a constant factor strictly less than $1$, which is due  
to what seems to be an inherently quantum phenomenon. 
We will use the geometry of the underlying interaction graph 
to achieve this separation, by treating 
differently expanding instances and non-expanding instances.
Before stating the main theorem of this section, 
we require a generalization of definition \ref{def:Lind}
and a simple fact. 

\begin{definition}
\textbf{$t$-independent set of constraints}\label{def:kind} 
\noindent 
Let $C$ be a quantum code with a set of $k$-local constraints, whose 
underlying bi-partite graph is $G(C)=(L,R;E)$. 
A set of constraints $U\subseteq R$ is said to 
be $t$-independent 
if for any $a,b\in U$ we have $\Gamma^{(2t+1)}(u) \cap \Gamma^{(2t+1)}(v) = \Phi$.
\end{definition}

The following fact can be easily derived by a greedy algorithm:

\begin{fact}\label{fact:indset}
Let $\eta = \eta(k,D_L) = k^{-(2k+1)} D_L^{-(2k-1)}$.
For any  
quantum code $C$ whose bi-partite graph $G(C)$ is left 
$D_L$-regular, and right $k$-regular, there exists a $k$-independent 
set of size at least $\eta n$. 
\end{fact}
\begin{proof}
Pick a constraint $u$, remove all constraints in 
$\Gamma^{(4k)}(u)$,  and repeat.
The number of constraints we have removed for each constraint is
$(kD_L)^{2k}$. Hence, we can proceed for $m/(kD_L)^{2k}$ steps.  
We get that the fraction of constraints is at 
least $k^{-(2k)}D_L^{-(2k)}$, 
and since $mk = n D_L$, we get the desired result. 
\end{proof}

\noindent 
{\bf Theorem (\ref{thm:sound})}
{\it Let $C$ be a stabilizer code on $n$ $d$-dimensional qudits, of minimal distance at least $k$, and a $k$-local ($k\geq 4$)
succinct generating set ${\cal G}\subset \Pi_d^n$, where the right 
degree of the interaction graph of ${\cal G}$ is $D_L$.
Then there exists a function $\gamma_{gap}=\gamma_{gap}(k)> min\left\{10^{-3},0.01/k \right\}$ 
such that for any
$\delta \leq \min\{dist(C)/2n,\eta/10\}$, 
(for $\eta$ as defined in Fact \ref{fact:indset})
we have 
$r(\delta')\leq \alpha(d) \left(1-\gamma_{gap}\right)$.
where $\delta' \in (0.99 \delta, 1.01\delta)$.
}

{~}

\noindent
The proof of the theorem will use, on one hand, claim (\ref{cl:QECC}) which upper-bounds the soundness of expanding instances, and on the other hand a lemma on non-expanding instances,which tries to "mimic" the behavior of the classical setting, in which non-expanding topologies suffer from poor soundness.
We now state this lemma:
\begin{lemma}\label{lem:indexp}
Let $C$ be a stabilizer code on $n$ 
qudits of dimension $d$, with minimal distance at least $k$
and a $k$-local ($k\geq 4$) succinct generating set ${\cal G}$, 
where the left degree of the interaction graph of ${\cal G}$ is $D_L$.
Let $\gamma_{gap} = \gamma_{gap}(k)=
min\left\{10^{-3},0.01/k \right\}$. 
If there exists a $k$-independent set $U$ of size $|U| = \delta n$,
with $\delta<dist(C)/2n$,
such that the bi-partite expansion error of $\Gamma(U)$ is at least $\eps = 0.32$, 
i.e. $|\Gamma(\Gamma(U))|= |\Gamma(U)| D_L (1-\eps')$ for some $\eps'\geq 0.32$ then
$$
r(\delta') \leq \alpha(d) \cdot (1-\gamma_{gap}),
$$
for some $\delta' \in (0.099 \delta,0.101\delta)$.
\end{lemma}

The proof of this Lemma is technically non-trivial, and we defer it to a 
separate section. 
From this lemma, it is easy to show theorem (\ref{thm:sound}):
\begin{proof}(of theorem \ref{thm:sound})
The parameters of the theorem allow us to apply directly fact (\ref{fact:indset});
hence there exists a $k$-independent set $S$ of size at least $\eta n$, for $\eta$ as defined in Fact (\ref{fact:indset}).
Since $\delta \leq \eta/10$ there exists a $k$-independent set $S$ of size $10 \delta$.
Now, either:
\begin{enumerate}

\item
$S$ has expansion error at least $0.32$.
By lemma (\ref{lem:indexp}), we have
$$
r(\mu) < \alpha(d)(1-\gamma_{gap}),
$$
for some $\mu \in (0.099 \cdot (10 \delta),0.101 \cdot (10\delta)) = (0.99 \delta,1.01\delta)$,
and $\gamma_{gap}(k)$ from lemma (\ref{lem:indexp}), 
which is at least $ min\left\{10^{-3},0.01/k \right\}$.

\item
The set $S$ is $\epsilon$-expanding for $\epsilon< 0.32$.  In which case,
since $S$ is in particular $R$-independent, then by claim (\ref{cl:QECC}), 
the soundness $r(\delta')\le 2\eps < 2/3 - 0.01 \leq \alpha(d) -0.01$, for all $\delta' \leq |S|/n$.  In particular $r(\mu) < \alpha(d)(1-0.01/k)$.

\end{enumerate}
Taking the higher of these two bounds we get the desired upper-bound for $r(\mu)$. 
\end{proof}

\subsection{Proof of Lemma (\ref{lem:indexp})}
In the following we first define the error; 
We provide the proof that the expected penalty of this error is small in fact
(\ref{fact:penalty}), then state and prove the Onion fact in sub-subsection 
\ref{sec:onion} and use it to prove Fact (\ref{fact:weight}),
in which we show that the error has large weight modulo the group.
Finally we combine all the above to finish the proof of the lemma. 

\subsubsection{Constructing the error}\label{sec:err}
Let $U\subseteq R$ be a $k$-independent set as promised by the conditions of the lemma.
Then $|U| = \delta n$, and denoting $S=\Gamma(U)$, we have that $|S| = \delta n k$.
Therefore, $|\Gamma(S)| = |S| D_L (1-\eps')$, for some $\eps' \geq 0.32$.
Let ${\cal E}$ be the following random error process: for each qudit of $S$ independently, we apply $I$ w.p. $1-p$ for $p=1/(10k)$, and one of the other elements of $\Pi_d$ with equal probability $p \cdot t(d)$, where $t$ is defined in 
Definition (\ref{def:single}).

$$
{\cal E} = \bigotimes_{i\in S} {\cal E}_i
\mbox{,  where  }
{\cal E}_i =
\left\{
	\begin{array}{ll}
		I_i  & \mbox{w.p. } 1-1/(10k) \\
		X_d^k P_d^l & \mbox{w.p. } t/(10k)
	\end{array}
\right.
$$
We note here that the choice of $p$ is such that on average, each $k$-tuple has only a small number of errors; the expectation of the number of errors is an absolute constant $1/10$ (not a fraction of $k$). 
This will help, later on, to lower-bound the weight of the error modulo the group.

\subsubsection{Analyzing Penalty}
We first claim, that on average, ${\cal E}$ has a relatively 
small penalty w.r.t. ${\cal G}$, using the fact that the expansion error 
is at least $0.32$ as in the condition of Lemma \ref{lem:indexp}.
For any ${\cal E}$, let $penalty({\cal E})$ denote the number of generators of ${\cal G}$ that do not commute with ${\cal E}$.
\begin{fact}\label{fact:penalty}
$$
\mathbf{E}_{\cal E}\left[Penalty({\cal E})\right] 
\leq 
p \alpha |S| D_L \left(1- 0.02/k \right)
$$
\end{fact}

\begin{proof}
Let $G=(L,R;E)$ 
denote the bi-partite graph corresponding to ${\cal G}$, 
with $R$ being the generators of ${\cal G}$ and $L$ the qudits. 
Let $S=\Gamma(U)$ be as before.
Let the error process ${\cal E}$ be the one defined above.
For any constraint $c\in \Gamma(S)$ which is violated when applied 
to this error, observe that there must be a qudit
$i\in supp(c)$ such that   
$\left[c|_i, {\cal E}_i\right] \neq 0$.  
We now would like to bound the number of constraints violated by ${\cal E}$ 
using this observation, and linearity of expectation. 

For an edge $e\in E$ connecting a qudit $i$ in $S$ and a constraint $c$ in 
$\Gamma(S)$,  
let $x(e)$ denote the binary variable which is $1$, iff the error 
term ${\cal E}_i$ on does not commute with $c|_i$.
In other words, an edge marked by $1$ is an edge whose qudit 
causes its constraint to be violated. 
By construction, for each $e\in E$ which 
connects the qudit $i$ and the constraint $c$ we have 
\begin{equation}
\label{eq:singleexp}
\mathbf{E}_{\cal E}[x(e)] = p(1-t).
\end{equation} 
This is true since a constraint $c$ restricted to the qudit $i$, 
$c|_i$ does not commute with the error restricted to the same qudit $i$,  
${\cal E}_i$, iff both ${\cal E}_i$ is non-identity (which happens with 
probability $p$) and
is not equal to $c|_i$. 

If we had just added now $x(e)$ over all edges going out of $S$ 
(whose number is $|S| D_L$), then by linearity 
of expectation, this would have given an upper bound on the expected 
number of violated constraint equal to 

\begin{equation}\label{eq:sumexp}
\sum_e p(1-t)=p |S| D_L \alpha(d). 
\end{equation}

Unfortunately this upper bound does not suffice; to strengthen it 
we would now like to take advantage of the fact that 
many of those edges go to the same constraint, due to the fact that the 
expansion is bad; thus, instead of simply summing these expectation values, 
we take advantage of the fact that two qudits touching the same 
constraint
cannot contribute twice to its violation. 
Observe that it may even be the case that some edges may cause 
constraints to become "unviolated", so the actual bound may be even lower.

Let $E_{inj} \subseteq E$ be a subset of the edges between $S$ to $\Gamma(S)$ 
chosen by picking a single edge for each constraint in $\Gamma(S)$. 
For an edge $e\in E$ let $c(e)$ denote the constraint incident on $e$, 
and let $e_{inj}(c(e))$ denote the edge
in $E_{inj}$ that is connected to $c(e)$.

We now bound the expectation by subtracting $x(e)$ from the sum, 
if the Boolean variable 
$x(e_{inj}(c(e)))$ is $1$; this avoids counting the violation of the same 
constraint twice due to the two edges. 
We have: 
$$
\mathbf{E}_{\cal E}\left[Penalty\right] 
\leq 
\mathbf{E}_{\cal E}\left[
\sum_{e\in E_{inj}} x(e) +
\sum_{e \notin E_{inj}} 
\left(1 - x(e_{inj}(c(e))) \right) \cdot x(e)
\right].
$$
Expanding the above by linearity of expectation:
$$
\mathbf{E}\left[Penalty\right] 
\leq 
\sum_{e\in E_{inj}} \mathbf{E}_{\cal E}\left[ x(e) \right] +
\sum_{e \notin E_{inj}} \mathbf{E}_{\cal E}\left[x(e)\right] - 
\sum_{e \notin E_{inj}} \mathbf{E}_{\cal E}\left[x(e_{inj}(c(e))) \cdot x(e) \right]=
$$
$$
\sum_{e\in E} \mathbf{E}_{\cal E}\left[ x(e) \right]+
\sum_{e \notin E_{inj}} \mathbf{E}_{\cal E}\left[x(e_{inj}(c(e))) \cdot x(e) \right]. 
$$

We have already calculated the first term in the sum in Equation 
\ref{eq:sumexp};  
We now lower bound the correction given by the second term. 
We use the fact that  for any $e\notin E_{inj}$
$$
\mathbf{E}_{\cal E}\left[x(e_{inj}(c(e)) x(e)\right] = 
\mathbf{E}_{\cal E}[x(e_{inj}(c(e)))]
 \mathbf{E}_{\cal E}[x(e)] $$
since ${\cal E}$ is independent between different qudits. 
We can thus substitute 
Equation \ref{eq:singleexp}, and get: 

$$
\mathbf{E}_{\cal E}\left[Penalty\right] 
\leq 
p \alpha |S| D_L  - |S| D_L \eps (p \alpha)^2.
$$
where we have used the fact that $|E \backslash E_{inj}| = |S| D_L \eps$.
This is equal to 
$$
p\alpha |S| D_L ( 1 - p\alpha \eps). 
$$
Using $p = 1/(10k),\eps \geq 0.32$, 
$\alpha(d) \geq 2/3$, we get the desired bound.
\end{proof}

\subsubsection{The Onion fact (\ref{fact:succinct})}\label{sec:onion}

\begin{fact}\label{fact:succinct}
\textbf{Onion fact}

\noindent
Let $C$ 
be a stabilizer code on $n$ qudits 
with a succinct generating set ${\cal G}$ of locality $k$, such that $dist(C)\geq k$.
Let $E\in \Pi_d^n$ s.t. 
$supp(E) \subseteq \Gamma(u)$ for some generator $u\in {\cal G}$. 
Finally let $\Delta\in A({\cal G})$, and let
$E_{\cal G}=\Delta\cdot E$. 
Then, for any $i\in [k]$, if 
$wt(E|_{\Gamma(u)}) = i$, then 
$wt(E_{\cal G}|_{\Gamma^{(2k+1)}(u)}) \geq min\left\{i,k-i\right\}$.
\end{fact}

\begin{proof}
If $\Delta|_{\Gamma(u)}=I$ then 
\begin{equation}\label{eq:i} 
wt \left(E_{\cal G}|_{\Gamma^{(2k+1)}(u)}\right) \ge 
wt \left(E_{\cal G}|_{\Gamma(u)}\right) =
wt \left(E|_{\Gamma(u)}\right) = i,
\end{equation}
so in this case we are done. 

Otherwise, $\Delta|_{\Gamma(u)}$ 
is non-identity, and so has at least one non-identity 
coordinate. Since $\Delta$ is non-identity,  
by the assumption on the succinctness of ${\cal G}$ we 
have $wt(\Delta) \geq k$.

Moreover, we claim that $wt \left(\Delta|_{\Gamma^{(2k+1)}(u)}\right) \geq k$.
Otherwise, consider the following process. 
Start with the generator $u$, and consider the qudits in $\Gamma(u)$. 
Now add the qudits in $\Gamma^{(3)}(u)$ (namely the qudits that are acted 
upon by generators intersecting $u$).   
Then add the next level, and so on for $k$ levels, 
by which point we have added all qudits belonging 
to $\Gamma^{(2k+1)}(u)$. 
By the pigeonhole principle, if $wt \left(\Delta|_{\Gamma^{(2k+1)}(u)}\right)<k$, 
then there must exist a level $t$, $1\le t \le k$, such that
$\Delta$ has zero support on qudits added in this level. 

We now claim that ${\tilde \Delta}=\Delta|_{\Gamma^{(2(t-1)+1)}(u)}$, 
is in the centralizer $\mathbf{Z}({\cal G})$ but its weight is less than $k$.
This, together with the fact that ${\tilde \Delta} \notin A({\cal G})$, 
shown in the next paragraph, contradicts the assumption that $dist(C)\geq k$.
To see that ${\tilde \Delta}$ is in the centralizer, 
we observe first that $\Delta$ commutes with all elements of ${\cal G}$ 
that act only on qudits in ${\Gamma^{(t-1)}(u)}$, 
and since ${\tilde \Delta}$ agrees with $\Delta$ on
${\Gamma^{(2(t-1)+1)}(u)}$, ${\tilde \Delta}$ also commutes with them. 
We also observe that ${\tilde \Delta}$ trivially commutes with all 
elements in ${\cal G}$ whose support does not intersect 
$\Gamma^{(2(t-1)+1)}(u)$. Hence we only need to worry about 
those terms that act on at least one qudit in 
$\Gamma^{(2t+1)}(u)-\Gamma^{(2(t-1)+1)}(u)$ 
and at least one qudit in  $\Gamma^{(2(t-1)+1)}(u)$.  
Let $v$ be some such term. 
Note that $v$ does not act on any qudit outside $\Gamma^{(2t+1)}(u)$ 
by definition.  
We know that $\Delta$ commutes with $v$.  
But by the choice of $t$, we know that $\Delta$ is trivial on those qudits 
added at the $t$-th level, and hence $\Delta$ restricted to 
$\Gamma^{(2t+1)}(u)$ (which contains the qudits of $v$)
is the same as $\Delta$ restricted to $\Gamma^{(2(t-1)+1)}(u)$. 
And so $\Delta$ restricted to $\Gamma^{(2(t-1)+1)}(u)$ commutes with $v$. 
\lnote{The above paragraph could be significantly shortened: 
By the property of $\Delta$ and 
definition of $\Gamma^{(i)}(u)$, we have that for any $g\in {\cal G}$
$[g|_B, \Delta|_B] = [g,\Delta]=0$, where $B = \Gamma^{(t-1)}(u)$.
Hence $[g,\tilde \Delta] = [g|_B, \tilde \Delta|_B] = [g|_B, \Delta|_B] = 0$.}

We showed that ${\tilde \Delta}$ is in $\mathbf{Z}({\cal G})$.
If it also belongs to $A({\cal G})$, this contradicts succinctness 
of ${\cal G}$;  
otherwise it is in $\mathbf{Z}({\cal G})-A({\cal G})$ 
implying the distance of $C$ 
is at most $k-1$, contrary to assumption.
This means that $wt \left(\Delta|_{\Gamma^{(2k+1)}(u)}\right) \geq k$.
Therefore, we now know 
by the triangle inequality on the Hamming distance, that 
\begin{equation}\label{eq:k-i}
wt \left(E_{\cal G}|_{{\Gamma^{(2k+1)}(u)}}\right) \geq 
wt \left(\Delta|_{{\Gamma^{(2k+1)}(u)}}\right) - 
wt \left(E|_{{\Gamma^{(2k+1)}(u)}}\right)
=
\end{equation}
$$
wt \left(\Delta|_{{\Gamma^{(2k+1)}(u)}}\right) - 
wt \left(E|_{\Gamma(u)}\right)
\geq
k-i.
$$
Taking the minimal of the bounds from Equations (\ref{eq:i}),(\ref{eq:k-i})
 completes the proof.
\end{proof}

\subsubsection{Analyzing error weight}
We note that the expected weight of ${\cal E}$ is 
$p|S|$ and since $|S|$ is linear in $n$, 
by Chernoff  
the probability that the weight of ${\cal E}$ is smaller by more than
than a constant fraction than this expectation is $2^{-\Omega(n)}$.
We need to show a similar bound on the weight modulo the centralizer group; 
given that $\delta<dist(C)/2n$ we only need to bound the weight modulo 
$A({\cal G})$.  
Let $\Delta\in A$ be some element in the stabilizer group 
and let ${\cal E}_{\cal G}=\Delta \cdot {\cal E}$. 
We now need to lower-bound $wt({\cal E}_{\cal G})$.

\begin{fact}\label{fact:weight}
For integer $k$, let $\hat{k} = \floor{k/2}+1$.
Let $y(k): [4,\infty] \mapsto \mathbf{R}$ be the function:
$$
y(k) =
\left\{
	\begin{array}{ll}
		1-2^{(-\hat{k}+1)log(k)+k-2.3\hat{k}+4.54}  & k\geq 12 \\		
   	     0.9999 & 6\leq k\leq 11 \\
		0.9992 & k=5 \\
		0.9985 & k=4 
	\end{array}
\right.
$$
We claim:
$$
Prob_{\cal E} 
\left( 
wt({\cal E}_{\cal G}) < 
|S| p y(k)  
\right)
= 
2^{-\Omega(n)}.$$
\end{fact}

\begin{proof}(\textbf{Sketch.The detailed proof can be found in Appendix 
(\ref{sec:weight}).})
The proof builds on the onion fact (\ref{fact:succinct}) as follows: the onion fact shows that "islands" with fewer than $k/2$ errors cannot "lose" error weight modulo the centralizer of ${\cal G}$.
The proof uses standard probabilistic arguments, 
to argue, that the random error pattern we chose, is such, that the vast majority of islands, have fewer than this threshold error weight, and so the overall error weight is
virtually unharmed.
\end{proof}

\subsubsection{Concluding the proof of lemma (\ref{lem:indexp})}
\begin{proof}
By fact (\ref{fact:penalty}) the average penalty of ${\cal E}$ is small, i.e. 
$$\mathbf{E}\left[Penalty({\cal E})\right] \leq |S| D_L p \alpha (1-0.02/k) 
\triangleq P.$$
Yet, by fact (\ref{fact:weight}) w.p. exponentially close to $1$, we have 
$$
wt({\cal E}_{\cal G}) 
\geq 
|S| p y(k) \triangleq W_{low}\geq |S| p \cdot 0.99.
$$
Similarly, by the Hoeffding bound w.p. exponentially close to $1$, we have
$$
wt({\cal E}_{\cal G}) < |S| p (1+0.01) \triangleq W_{high}.
$$

Since all penalties are non-negative, we conclude that {\it conditioned} on
$\left|wt({\cal E}_{\cal G})/ (|S|p) -1\right| < 0.01$, we have
$\mathbf{E}\left[Penalty({\cal E})\right] \leq P+2^{-\Omega(n)}$.
Therefore, there must exist an error ${\cal E}$, whose weight modulo
 ${\cal G}$ deviates
by a fraction at most $0.01$ from $|S| p$, and whose penalty is at most $P+2^{-\Omega(n)}$.

We would like to bound the soundness of this error, which is the ratio 
of the penalty to its relative weight times $D_L$. 
We get that its soundness is at most 
\begin{equation}\label{eq:sound}
r = \frac{P+2^{-\Omega(n)}}{D_L W_{low}}\leq 
\frac{1}{D_L}
\cdot
\frac
{|S| D_L p \alpha (1-0.019/k)}
{|S| p y(k)}
=
\alpha \left(\frac{1-0.019/k}{y(k)}\right).
\end{equation}

We now note that in the last expression,  for all $k\geq 12$, 
the ratio $\frac{1-0.019/k}{y(k)}$ 
is at most $1-0.01/k$. 
For all values of $4 \leq k < 12$ we substitute the appropriate 
value of $y(k)$ and get similarly that 
the ratio $\frac{1-0.019/k}{y(k)}$ is at most $1- 10^{-3}$. 
Hence, the soundness of the error, $r$ is at most $\alpha(d)(1-\gamma_{gap})$ 
where $\gamma_{gap}$ is as defined in the statement of theorem 
(\ref{thm:sound}). 

\end{proof}

\section{Acknowledgements}
The authors would like to thank 
Eli Ben-Sasson, Irit Dinur and Tali Kaufman for insightful discussions.

\begin{appendix}

\section{Proof of Claim \ref{cl:defequiv}}

We prove that 
Definitions (\ref{def:det2}) and (\ref{def:det1}) are equivalent:

\begin{proof}
If definition (\ref{def:det1}) holds then for any $E\in {\cal E}$, and any two orthogonal states of the code $\ket{\phi},\ket{\psi}$ have
$$ 
\langle \phi |E| \psi \rangle = 
\langle \phi |\Pi_C E \Pi_C| \psi \rangle = 
\gamma_E \langle \phi |\Pi_C| \psi \rangle = 
\gamma_E \langle \phi |\psi \rangle =
\gamma_E 0
= 0.
$$
On the other hand, suppose that for any two orthogonal states $\ket{\phi},\ket{\psi}$ in the code, and any $E\in {\cal E}$, we have $\langle \phi |E| \psi \rangle = 0$.
Choose some orthogonal basis of the code $C$  $\left\{ \ket{b_i} \right\}_{i=1}^m$.
Then for each of these basis vectors, we have $\langle b_i |E| b_j \rangle=0$, for $i\neq j$.
Hence, in particular, the operator $E|_C$, i.e., $E$ restricted to $C$, is a diagonal matrix $diag(\lambda_1,\hdots,\lambda_m)$.
We claim, further that $E|_C = \gamma_E I$, for some constant $\gamma_E$, and hence $\Pi_C E \Pi_C = \gamma_E \Pi_C$.
Suppose, on the negative, that there exist two eigenvalues of $E|_C$ that are different, say $\lambda_1 \neq \lambda_2$.
Consider the orthogonal states
$\ket{\phi} = \frac{1}{\sqrt{2}}\left(\ket{b_1}+\ket{b_2}\right)$,
$\ket{\psi} = \frac{1}{\sqrt{2}}\left(\ket{b_1}-\ket{b_2}\right)$.
Then $\ket{\phi},\ket{\psi}$ are in the code by linear closure, and are orthogonal, and yet 
$$
\langle \phi |E| \psi \rangle =
\frac{1}{2}\langle b_1 | E| b_1 \rangle  - 
\frac{1}{2}\langle b_1 | E| b_2 \rangle + 
\frac{1}{2}\langle b_2 | E| b_1 \rangle - 
\frac{1}{2}\langle b_2 | E| b_2 \rangle =
\frac{1}{2}(\lambda_1 - \lambda_2) \neq 0,
$$
contrary to our assumption on $E$. 
\end{proof}

\section{Proofs of geometrical facts on small-set expanders}\label{sec:bipartite}

\subsection{Proof of fact (\ref{fact:essence}):}

For $S\subseteq R$ let $\Gamma_1(S)\subseteq \Gamma(S)$ denote the subset of the neighbors of $S$ with exactly one neighbor in $S$.
Similarly, let $\Gamma_{\geq 2}(S)$ denote
the subset of neighbors with at least two neighbors in $S$.
\begin{proof}
The average degree of a vertex in $\Gamma(S)$ w.r.t. $|S|$
is at most $\frac{D_L S}  {D_L S (1-\eps)} = \frac{1}{1-\eps}$.
Let $\alpha_1$ denote the fraction $|\Gamma_1(S)| / |\Gamma(S)|$, 
where $\Gamma_1(S)$ is the set of neighbors of $S$ with degree exactly 
$1$ with respect to $S$. 
Then 
$$\frac{1}{1-\eps} \ge \alpha_1 1 + (1-\alpha_1) m,$$
where $m$ is the average degree of a vertex with at 
least two neighbors in $S$.
Then by simple algebra

$$\alpha_1(m) \ge 
1 - \frac{1}{m-1}\cdot\frac{\eps}{1-\eps},   
$$
so $\alpha_1(m)$ is a monotonously increasing function of $m$,
and since $m\geq 2$, then $\alpha_1$ is minimized for $m=2$.
Hence, 
$$ \alpha_1 \ge  1-\frac{\eps}{1-\eps}. $$
and since $\eps<1/2$ we have:
$$ \alpha_1 \geq 1 - \eps(1+2\eps) \geq 1-2\eps.$$
\end{proof}

\subsection{Proof of fact (\ref{fact:deg}):}

\begin{proof}
By definition, we have $|\Gamma(S)| \geq |S| D_L (1-\eps)$.
Let $E_{inj}\subseteq E(S)$ be a subset of the edges incident on $S$ such that each $u\in \Gamma(S)$ has 
a single neighbor in $S$ connected by an edge of $E_{inj}$.
Then $E_{inj}$ is of size $|\Gamma(S)|$ which is at least $|S| D_L (1-\eps)$.
Also $|E(S)| = |S| D_L$, thus $|E(S)-E_{inj}| \leq |S| D_L \eps$.
Therefore $\left| \Gamma_{\geq 2} (S) \right| \leq |S| D_L \eps$.
Hence, $\Gamma_1(S) = \Gamma(S)-\Gamma_{\geq 2}(S)$ is of size at least $|S| D_L (1-\eps) - |S| D_L \eps = |S| D_L (1-2\eps)$.
Therefore, when $\eps<1/2$
there exists a vertex $v\in S$ with at least $D_L (1-2\eps)$ neighbors in $\Gamma_1(S)$.
Since $v$ has $D_L$ neighbors in $\Gamma(S)$, then the fraction of neighbors of $v$ with at least two neighbors in $S$ is at most $2\eps$, when $\eps<\frac{1}{2}$.
\end{proof}

\section{Existence of arbitrarily sound 
classical $\LTC$s on small-set expanders}\label{app:classical}
\begin{claim}\label{cl:classical}
For any $\eps\in (0,1/2)$, and $r\in (0,1)$ there exists 
a constant $\delta=\delta(r,\epsilon)$, such that there exists
an explicit infinite 
family of codes $\left\{C_{\eps}(n)\right\}_{n\in \mathbf{N}}$, of $n$ bits,
of constant fractional rate $r$, and constant fractional distance 
$d=d(\epsilon,r)$, 
whose check terms are of locality whose {\it expectation}
is equal to a constant $k$, 
and all errors of weight less than $\delta n$ have soundness 
$r(\delta)\geq 1-3\eps$.
Moreover, the underlying graph of these codes is an $\eps$ small-set expander.
\end{claim}

\begin{proof}
The construction of $\cite{CRVW}$, generates explicitly
for any $\eps,r$ a left-$D_L$-regular bi-partite 
graph $G=(L,R;E)$ such that $|R|/|L|=1-r$,  
and for any subset $S\subseteq L$, $|S|\leq |L| \delta$ the neighbor set of $S$ is of size at least $|S| D_L (1-\eps)$, where $D_L$ is the left degree of $G$.
Note that since the left degree is $D_L$, the average right degree is 
$D_L|L|/|R|=D_L\frac{1}{1-r}$, which is a constant given that $D_L$ is 
a constant. 

The code is defined by assigning to each right node a 
parity check over its incident vertices.  
Let us lower bound the rate of this code: it is at least 
$r=(|L|-|R|)/|L|$, since each constraint in $R$  at most halves the dimension 
of the codespace. The minimal distance of the code is at least $\delta$, 
since any non-zero word of weight at most $\delta$ is rejected, 
since there exists at least one check term
that "sees" just a single bit at state $1$, by Fact (\ref{fact:essence}). 

Hence, these are so-called "good" codes. 
Furthermore, their soundness is at least $1-3\eps$ since 
an error on a set of bits $S$ of size $|S|\leq\delta n$, 
is examined by at least $|S|D_L(1-\epsilon)$ constraints. 
By Fact (\ref{fact:essence}) at least $1-2\epsilon$ of those constraints, examine $S$ 
in exactly one location; all 
constraints that touch a given error set $S$ in exactly one location 
will be violated; hence the total number of constraints that will be 
violated is at least $|S| D_L (1-\eps) (1-2\eps)\geq |S| D_L (1-3\eps)$. 
Therefore, the soundness function $R(\delta')$ is at least $(1-3\eps)\delta' k$, 
for all $\delta'\in [0,\delta]$.
\end{proof}

\section{Proof of Lemma (\ref{lem:Pauli}): decomposition to
cosets of a stabilizer code}\label{sec:lemPauli}

\begin{proof}

For any $E\in \Pi_d^n$, and any $g\in {\cal G}$, we have
$E g = \omega g E$, where $\omega \in \mathbf{C}$.
Therefore, for any $\ket{\eta}$ in $C$, we have $E\ket{\eta}$ is an $\omega$ eigenstate of $g$.
Then for any $E\in \Pi_d^n$, we have that $E C$ is some simultaneous eigenspace of ${\cal G}$.
But, since $\Pi_d^n$, spans over $\mathbf{C}$ all unitaries on $n$ qudits, then it must be that every simultaneous eigenspace of ${\cal G}$
is equal to $EC$ for some $E\in \Pi_d^n$.
In particular, any state 
$\ket{\phi}$ may be written as a sum 
$$ \ket{\phi} = \sum_i E_i \ket{\eta_i},$$
where $E_i\in \Pi_d^n$, and $\ket{\eta_i}\in C$.
\end{proof}

\section{Proof of Claim (\ref{cl:distequiv}) Equivalence of definitions of code distance}\label{sec:distequiv}

We prove that a
stabilizer code $C$ has $dist(C)\geq \rho$ by definition \ref{def:stabdist}, 
iff it has distance $\geq\rho$ by definition \ref{def:qeccdist}. 

\begin{proof}
If the minimal weight of a Pauli in $\mathbf{Z}({\cal G})-A({\cal G})$ has weight at least $\rho$,
then all terms $E\in \Pi_d^n$ of weight strictly less than $\rho$ 
(namely, at most $\rho -1$) are either
spanned by ${\cal G}$, or outside $\mathbf{Z}({\cal G})$.
Take any two orthogonal code states $\ket{\phi},\ket{\psi}$.
If $E\in A({\cal G})$ then all code states are stabilized by $E$, so we have
$\langle \phi | E | \psi \rangle = 1 \cdot \langle \phi | \psi \rangle = 0$.
If $E\notin \mathbf{Z}({\cal G})$, $E$ does not commute with some generator, so in particular, $E$
does not preserve the simultaneous $1$-eigenspace of all generators, namely, 
the code. 
By lemma (\ref{lem:Pauli}), this implies that $E C$ is orthogonal to $C$.
Thus we have in this case as well : $ \langle \phi| E |\psi \rangle = 0$.
Hence the minimal distance of the code, according to definition (\ref{def:qeccdist}) is at least $d$.

Proving the converse, assume that $dist(C)<\rho$, i.e. $min_{E\in \mathbf{Z}({\cal G})-A({\cal G})} wt(E)<\rho$.
Then, there exists $E\in \Pi_d^n$, of weight less than $\rho$, that commutes with all generators of ${\cal G}$ but not spanned by them,
so there exists some state $\ket{\phi}\in C$, such that $E\ket{\phi} \neq \ket{\phi}$, yet $E \ket{\phi}\in C$,
(see \cite{Got}, p. 27).
Thus, there exists a non-zero projection of $E \ket{\phi}$ on some other code state $\ket{\psi}$ orthogonal to $\ket{\phi}$.
Therefore, $\langle \psi | E |\phi \rangle \neq 0$, contrary to definition (\ref{def:det2}).
\end{proof}

\section{Lower-bound on weight: proof of Fact (\ref{fact:weight})}
\label{sec:weight}

\begin{proof}
Let $x \sim B(k,p=1/(10k))$ denote a random variable which is the sum of $k$ 
i.i.d Boolean variables, each equal to $1$ with probability $p$; 
in other words, $x$ is a 
binomial process; $B(i) = Prob(x=i)$.
Let $U$ be a $k$-independent set of size $\Omega(n)$, and ${\cal E}$ be 
the error process defined in Subsection (\ref{sec:err}). 
Let $U_i = \left\{u\in U | wt({\cal E}|_{\Gamma(u)}) = i \right\}$ be the 
set of generators which have exactly $i$ erroneous qudits. 
Using the Hoeffding bound, for a given $i\in [k]$ and a given any constant
$\chi>0$, we have 
\begin{equation}
Prob_{\cal E} \left( \left|\frac{|U_i|}{|U|}-B(i)\right|\geq \chi\right) = 2^{-\Omega(n)}
\end{equation}
By the union bound, we have that for any constant $\chi>0$:
\begin{equation}\label{eq:qecc3}
Prob_{\cal E}\left(\exists i, s.t. \left|\frac{|U_i|}{|U|}-B(i)\right|\geq \chi\right) = 2^{-\Omega(n)}.
\end{equation}
\noindent
Since the set $U$ is a $k$-independent set, then the sets $\left\{\Gamma^{(k)}(u)\right\}_{u\in U}$ are non-intersecting so
\begin{equation}\label{eq:qecc1}
wt({\cal E}_{\cal G}) \ge \sum_{u\in U} wt \left({\cal E}_{\cal G}|_{\Gamma^{(k)}(u)}\right),
\end{equation}
By the Onion fact (Fact \ref{fact:succinct}),
for each $u\in U_i$ we have
$wt \left({\cal E}_{\cal G}|_{\Gamma^{(k)}(u)}\right)\geq min\left\{i,k-i\right\}$, 
hence
$$
wt({\cal E}_{\cal G})
\ge
\sum_{i\in [k]}
|U_i| min\left\{i,k-i\right\}
=
\frac{|S|}{k}
\sum_{i\in [k]}
 \frac{|U_i|}{|U|}min\left\{i,k-i\right\}
$$
using $k|U|=|S|$. Using equation (\ref{eq:qecc3}) w.p. close to $1$ we have 

\begin{equation}\label{eq:qecc4}
wt({\cal E}_{\cal G})\ge \frac{|S|}{k} 
\sum_{i\in [k]} (B(i)- \chi) min \left\{i,k-i\right\} 
\geq 
\frac{|S|}{k}
\left(
\sum_{i\in [k]} B(i) min\left\{i,k-i\right\} - 2^{-k^2}
\right),
\end{equation}
for $\chi = 2^{-k^2}/k^2$.

We separate the rest of the proof to two cases: 
$k\ge 12$ and $4\le k< 12$.  
We start with the case $k\ge 12$. 
Recall $\hat{k} = \floor{k/2}+1$.  
Let
$$ A_{loss} = \sum_{i \geq \hat{k}} B(i) (2i - k).$$
Then by equation (\ref{eq:qecc4}) we have that with probability 
exponentially close to $1$ 
\begin{equation}\label{eq:qecc6}
wt({\cal E}_{\cal G}) 
\geq
\frac{|S|}{k}
\left( 
\sum_{i\in [k]} B(i) i
- 
A_{loss} - 
2^{-k^2} 
\right)
= 
\frac{|S|}{k} 
\left(
 p k - 
2^{-k^2} - 
A_{loss} 
\right)
\end{equation}

In the rest of the proof for $k\ge 12$ 
we upper-bound $A_{loss}$ and substitute 
in the above equation to derive the desired result.
Using an upper-bound of the binomial, we have: 
\begin{equation}\label{eq:binom}
B(\hat{k}) = 	{k \choose \hat{k}} p^{\hat{k}}(1-p)^{\hat{k}}
\leq 
2^{k} \cdot (10k)^{-\hat{k}} (1-p)^{\hat{k}}
\leq
k^{-\hat{k}} 10^{-\hat{k}}  2^k
\leq
2^{-\hat{k}log(k)+k-3.3\hat{k}},
\end{equation}
For any $i\geq \hat{k}$ and $p<1/2$ we have
\begin{equation}\label{eq:binbound}
B(i+1) = B(i) \left(\frac{k-i}{i+1}\right) \left(\frac{p}{1-p}\right) < B(i) \frac{p}{1-p} < 2 p B(i)
\end{equation}
Substituting equations (\ref{eq:binbound}) and (\ref{eq:binom}) in the expression for $A_{loss}$ we have:
\begin{equation}\label{eq:qecc5}
A_{loss} =
\sum_{i \geq \hat{k}}^k B(i) (2i-k)
\leq
2^{-\hat{k}log(k)+k-3.3\hat{k}} 
\sum_{i \geq {\hat{k}}}^k 
(2p)^{(i-\hat{k})} (2i-k)
\end{equation}
\begin{equation}
\leq
2^{-\hat{k}log(k)+k-3.3\hat{k}+1+\hat{k}} 
\sum_{i \geq {\hat{k}}}^k 
(p)^{(i-\hat{k})} (i-\floor{k/2})
\end{equation}
Changing summation $i - \floor{k/2} \mapsto j$ we have the above is at most:
\begin{equation}
2^{-\hat{k}log(k)+k-2.3\hat{k}+1} 
\sum_{j \geq 1}^{{\ceil{k/2}}} p^{-j+1} j
\leq
2^{-\hat{k}log(k)+k-2.3\hat{k}+1} 
\sum_{j \geq 1}^{{\ceil{k/2}}} p^{-j+1} k
\end{equation}
\begin{equation}
\leq
2^{-\hat{k}log(k)+k-2.3\hat{k}+1} 
 k 
\sum_{j \geq 1}^{{\ceil{k/2}}} p^{-j+1}
\leq
2^{-\hat{k}log(k)+k-2.3\hat{k}+1}
 k \cdot 1.1
\leq 
2^{(-\hat{k}+1)log(k)+k-2.3\hat{k}+1.2},
\end{equation}
where in the last inequality we bound the sum by $\sum_{i\geq 0} 1/p^i$,
and set $p = 1/(10k) \le 1/100$, using $k\ge 12$. 
Substituting this value in (\ref{eq:qecc6}) we have that with probability 
$2^{-\Omega(n)}$ close to $1$, 
$$
wt({\cal E}_{\cal G}) \geq 
\frac{|S|}{k}
\left( 
pk
- 
2^{-k^2} 
- 
2^{(-\hat{k}+1)log(k)+k-2.3\hat{k}+1.2} 
\right)
=
$$
$$
\geq
\frac{|S|}{k}
\left( 
pk
- 
2^{(-\hat{k}+1)log(k)+k-2.3\hat{k}+1.21} 
\right)
$$
where in the last inequality we used again $k\ge 12$. 
Continuing, using $p=\frac{1}{10k}$ the above bound is equal to 
$$
=|S| p \left(1 - 2^{(-\hat{k}+1)log(k)+k-2.3\hat{k}+1.21+log_2(10)}\right)
\geq 
|S| p y(k),
$$
for all $k\geq 12$.
For values of $4\leq k < 12$ we substitute directly $k$ in Equation (\ref{eq:qecc4}), evaluate, and show it is at least $|S| p y(k)$.
\end{proof}

\section{Quantum $\PCP$ of Proximity}\label{app:PCPP}

\subsection{Classical $\PCP$s of Proximity}
We begin by presenting the definitions following \cite{BGHSV}. 
A pair language $L$ is a subset of 
$\left\{0,1\right\}^n\times \left\{0,1\right\}^\ell$ for $\ell = poly(n)$.
For a pair language $L$, let 
$L(x) =\left\{y | (x,y)\in L\right\}$.

\begin{definition}

\textbf{$\PCP$ of proximity ($\PCPP$)}

\noindent
For functions $s,\delta: Z^+ \mapsto [0, 1]$, a verifier $V=V(x)$ is a probabilistically checkable proof of
proximity ($\PCPP$) system for a pair language $L$ 
with proximity parameter $\delta$ and soundness error $s$ if the
following two conditions hold for every pair of strings $(x,y)\in 
\left\{0,1\right\}^n\times \left\{0,1\right\}^\ell$:

\begin{enumerate}
\item
Completeness: If $(x,y)\in L$ there exists $\pi$ such that $V(x)$ accepts oracle $y \circ \pi$ with probability $1$.
\item
Soundness: If $y$ is $\delta(|x|)$-far from $L(x)$, then for every $\pi$, the verifier $V(x)$ accepts oracle $y\circ \pi$ with
probability at most than $s(|x|)$.
\end{enumerate}
If $s$ and $\delta$ are not specified, 
then both can be assumed to be constants in $(0,1)$.

\end{definition}

The {\it query complexity} of the verifier $V$ is defined to be the number of 
coordinates that $V$ queries out of $y$ and $\pi$. $V$ is not charged 
for reading $x$ but is charged for reading $y$ even though it is part of 
the input. We notice that this is a more stringent restriction that 
in the case of a $\PCP$ proof; however, the requirements on 
the proof system are weaker - $V$ is supposed to reject only word which 
are {\it far} from words in the language. 

A good pair language to keep in mind is CIRCUIT-VAL, i.e. 
the pairs $(x,y)$, where $x$ is a circuit
on $n$ bits of polynomial size, and $y$ is a string on $n$ bits, 
and $(x,y)\in L$ if $x(y) = 1$, i.e. the circuit
$x$ evaluates to $1$ on input $y$. 
Though this problem lies in $\P$, a
simple argument (Proposition 2.4 in \cite{BGHSV}) shows that 
a $\PCPP$ for CIRCUIT-VAL, implies a 
$\PCP$ for the $\NP$ complete decision language CIRCUIT-SAT, 
the set of all $x$, for which there exists $y$, such that $x(y)=1$.

\subsection{From $\PCPP$s to $\LTC$s}
Given a $\PCPP$, \cite{BGHSV} provides 
a standard construction of an $\LTC$ with related parameters, as follows. 
Given is a $\PCPP$ for membership in a code, 
namely, for the pair language of $(C,w)$, a code and a member in that code. 
Suppose the proximity parameter of the $\PCPP$ is $\delta$, the soundness $s$ 
and the query complexity $k$. Suppose also that we are given a code $C$ 
with distance $D$. 
Then, one can construct an error correcting code $C'$ 
which is an $\LTC$   
with $k$-local constraints, whose weighted distance is $D$, 
and whose soundness is proportional to the soundness $s$. 

$C'$ is defined as the strings $w \circ \pi$ for all $w$ in $C$, 
where $\pi$ is the proximity proof of $w$. 
If one defines the distance by weighting only the coordinates in the first 
register, then $C'$ trivially has the same distance as $C$\footnote{In \cite{BGHSV} this choice of definition of distance is referred to 
as equivalent to the one used in \cite{BGHSV}, in which 
many repetitions of the string $w$ are taken, so that the weight of 
the error on the second, proof, register becomes negligible.} 
The local test for the code $C'$ as an $\LTC$ are the 
$k$-local tests performed by the verifier of the $\PCPP$;
Consider now a word $w'\circ \pi'$ which
 is $\delta$-far from any $w\circ \pi$ in the code 
$C'$, where the distance is measured again by taking into account only the 
coordinates of the left register. This means that $w'$
is $\delta$-far from a word in the code 
$C$, then the tests will reject the word $w'\circ \pi'$ with probability 
$s$, which will thus be the soundness of the code for proximity $\delta$. 

\subsection{Quantum $\PCPP$s}
We now define the quantum analogue of $\PCP$'s of proximity.
We consider quantum pair languages 
$L \subseteq \left\{0,1\right\}^n \otimes {\cal H}_{prf}$, 
where ${\cal H}_{prf}$ is a Hilbert 
space of $\ell$ $d$-dimensional qudits, for $\ell=poly(n)$.
For a quantum pair language $L$ 
let $L(x) = Span\left\{\ket{\psi}\in {\cal H}_{prf}, 
(x,\ket{\psi})\in L\right\}$.
 
\begin{definition}

\textbf{Quantum $\PCP$ of proximity}

\noindent
Fix functions $s,\delta: Z^+ \mapsto [0, 1]$. Let  
$V=V(x)$ be a function from $n$ bit strings $x$
to sets of $m$ $k$-local projections $\{\Pi_i\}_{i=1}^m$, 
each acting on ${\cal H}_{prf} \otimes{\cal H}_{pxmty}$. $V$ is a 
quantum probabilistically checkable proof of
proximity ($\qPCPP$) system, for a quantum pair language $L$, 
with proximity parameter $\delta$ and soundness error $s$, if the
following two conditions hold for every pair $(x,|\psi\ra)$: 

\begin{enumerate}
\item
Completeness: 
If $(x,\ket{\psi})\in L$, 
there exists $\ket{w}\in {\cal H}_{pxmty}$, 
such that for all check terms $\Pi_i\in V(x)$  
$$ \Pi_i \left(\ket{\psi} \otimes \ket{w}\right) = 0.$$
\item Soundness: 
If $\ket{\phi}$ is a quantum state in 
${\cal H}_{prf} \otimes{\cal H}_{pxmty}$ 
whose reduced density matrix to ${\cal H}_{prf}$ 
is supported on states, each
of distance at least $\delta(|x|)$ from $L(x)$, then
$$ \frac{1}{m}\sum_i \langle \phi | \Pi_i |\phi \rangle \ge s(|x|):.$$
\end{enumerate}

\end{definition}

\subsection{From $\qPCPP$s to $\qLTC$s}

Given is a $\qPCPP$ for membership in a quantum code on $\ell$ qudits,  
namely, for the pair language $L$ comprised of pairs $(C,\ket{\psi})$: a 
code (described by $n$ bits), and an $\ell$ qudit state in the code. 
Suppose $L$ has a $\qPCP$ of proximity with parameters 
$\delta,s$ for some 
functions $s,\delta: Z^+ \mapsto [0, 1]$, with projections $\Pi_i$.  
Let $C'$ be the codespace $\subseteq 
{\cal H}_{prf}\otimes {\cal H}_{pxmty}$, defined as: 
$$Span\left\{ \ket{\phi}\otimes \ket{\Pi(\phi)}\mbox{  s.t.  } \ket{\phi}\in C\right\},$$
where $\ket{\Pi(\phi)}$ is some proof of proximity for $\ket{\phi}$ from the 
$\qPCPP$. 
Let $dist_{prf}$ denote the distance from the codespace as 
in definition (\ref{def:distcode}) 
except it only counts non-identity Paulis acting on the 
register ${\cal H}_{prf}$. 

\begin{claim}
$C'$ is a $\qLTC$ with query complexity $k$ and 
soundness $R(\delta)= s$ (where the proximity $\delta$ is
defined with respect to the distance $dist_{prf}$).
\end{claim}

\begin{proof}
Set $\{\Pi_i\}_{i=1}^m$ as the check terms for $L(C)$.  
These are $k$-local terms, so $C'$ has query complexity $k$.
By definition of the quantum $\PCP$ of proximity, for 
any state $\ket{\phi}$ in the codespace of $C'$, we have 
$\Pi_i \ket{\phi}=0$, for any $\Pi_i\in V(C')$.
Let us assume now that $dist_{prf}(\ket{\phi},C')\geq \delta(|C|)\cdot \ell$.
Then by Definition \ref{def:distcode}, for any Pauli operator $E$ acting on 
${\cal H}_{prf}\otimes {\cal H}_{pxmty}$, whose support on 
${\cal H}_{prf}$ is at most $\delta(|C|)\cdot \ell - 1$, 
we have that $E\ket{\phi}$ is still orthogonal to $C'$.
In particular, for any $E$ whose support is 
contained in ${\cal H}_{prf}$, and whose weight is at most 
$\delta(|C|)\cdot \ell-1$, we have
that $E\ket{\phi}$ is still orthogonal to $C'$.
It is easy to see that the reduced state of $\ket{\phi}$ to 
${\cal H}_{prf}$ is a mixture of orthogonal 
states $\left\{\ket{\eta_i}\right\}_i$, each of which is at least 
$\delta(|C|)\cdot \ell$-far from $C$.
By virtue of the soundness of the $\qPCPP$, $\ket{\phi}$ 
will be rejected by the check terms $\Pi_i$ with probability at least $s(|C|)$.
\end{proof}

\end{appendix}

\end{document}